\documentclass[aps,pra,twocolumn,superscriptaddress]{revtex4-2}

\usepackage{graphicx}
\usepackage{amsmath,amssymb}
\usepackage{amsthm}
\usepackage{bm}
\usepackage{xurl}   

\newtheorem{lemma}{Lemma}[section]

\newcommand{\ket}[1]{\vert#1\rangle}
\newcommand{\bra}[1]{\langle#1\vert}

\newcommand{\ii}{\mathrm{i}}
\newcommand{\dd}{\mathrm{d}}
\newcommand{\iz}{\mathrm{i}0^+}

\begin{document}

\title{Quantum scattering by intersecting $\delta$-potential barriers:\\
from Gaudin's kaleidoscope to quantum Galperin billiards}
\author{Yi-Cong Yu}
\affiliation{Wuhan Institute of Physics and Mathematics, Innovation Academy for Precision Measurement Science and Technology, Chinese Academy of Sciences, Wuhan 430071, China}
\author{Wen-Jie Qiu}
\affiliation{Wuhan Institute of Physics and Mathematics, Innovation Academy for Precision Measurement Science and Technology, Chinese Academy of Sciences, Wuhan 430071, China}
\author{Xiaoming Cai}
\email{cxmpx@wipm.ac.cn}
\affiliation{Wuhan Institute of Physics and Mathematics, Innovation Academy for Precision Measurement Science and Technology, Chinese Academy of Sciences, Wuhan 430071, China}
\date{\today}

\begin{abstract}
We develop a general framework for the scattering of a plane wave by a class of
Gaudin kaleidoscope models---straight $\delta$-potential barriers intersecting at
a common point. The projected Lippmann--Schwinger equations split into a singular
part, a finite pole closure generated by a geometric moving rule, and a regular
remainder, the outgoing state being an atomic measure on the circle. At the
special angles $\pi/N$, where the intersecting barriers generate the dihedral
group $D_N$, the number of channels stays fixed, whereas at generic angles
channels are created and destroyed, opening smoothly from zero weight as a
critical direction is crossed. This leaves no singular trace in the probabilities, being
carried instead by the phase shifts---a quantum--classical correspondence beyond
the reach of the coordinate Bethe ansatz. The same equations solve the quantum
Galperin billiards exactly, the method of images reducing them to a single
linear system whose phase shifts follow in closed form.
\end{abstract}

\maketitle

\section{Introduction}

Gaudin observed that the Bethe wave function of one-dimensional particles with repulsive $\delta$ interactions is a coherent superposition of reflected waves of a generalized kaleidoscope, with mirrors generating a finite reflection group \cite{bethe1931,yang1967,yang1968,lieb1963,lieb1963b,gaudin2014bethe, coxeter1973regular,humphreys1992reflection}. The associated Gaudin structures remain central to contemporary integrable models \cite{buric2021gaudin}. We recently encoded the consistency of this coordinate Bethe ansatz (CBA) in the kaleidoscope Yang--Baxter equation (KYBE), the statement that multiple scattering closes around the center \cite{qiu2026}. The KYBE is a necessary condition for solvability. For intersecting mirrors it becomes concrete: those meeting at the dihedral angle $\pi/N$ must carry equal couplings \cite{jackson2024}.

Solvability is not settled by the KYBE alone. The mass-imbalanced two-body problem in a hard-wall trap is exactly solvable although it violates these mirror conditions \cite{lqzc2019}. The asymmetric Bethe ansatz (ABA) is equally telling: it keeps exact solvability when part of the mirrors become fully reflecting, yet only inside the symmetry subspace selected by the boundary \cite{jackson2024,olshanii2025}. Solvability thus depends not only on the KYBE, but also on the boundary conditions and the symmetry sector on an equal footing.

In solving the Bethe equations for the bounded kaleidoscope, we encountered a purely technical restriction: the problem is solvable only for a few values of $N$ \cite{qiu2026}. This restriction has no classical geometric counterpart, since every angle $\pi/N$ is equally admissible as the opening angle of a finite reflection group \cite{richens1981,zemlyakov1975,ott2002}. Rather, it arises from the boundary \cite{qiu2026}. Removing the boundary reduces the problem to its simplest symmetric scattering setting, governed solely by reflections in the barriers. We therefore consider a plane wave scattered by straight $\delta$-potential barriers intersecting at a common point (Fig.~\ref{fig:schematic}). This model is the main focus of the present work.

\begin{figure}[t]
\centering
\includegraphics[width=0.95\columnwidth]{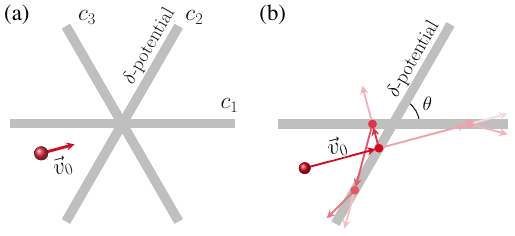}
\caption{
Schematic of the models considered in this work.
(a) General configuration of $\delta$-potential barriers intersecting at a common vertex, with $c_m$ the coupling strength of the $m$th barrier and $\vec v_0$ the incident velocity. Three equally spaced barriers are shown only for clarity; their number and relative angles are arbitrary.
(b) Two-barrier configuration used to analyze the essential scattering mechanism. The barriers intersect at an angle $\theta$, and each encounter splits a ray into reflected and transmitted branches, whose relative weights are indicated schematically by color intensity. Rays passing on opposite sides of the vertex follow different scattering sequences, while the trajectory aimed directly at the vertex is singular. This minimal configuration captures the essential structure of the general problem.
}
\label{fig:schematic}
\end{figure}

The contact-interacting one-dimensional gases in which these structures arise are themselves an active platform, with exact solutions reviewed in Ref.~\cite{minguzzi2022exact} and integrable dynamics probed directly in experiment \cite{senaratne2022spincharge}. Their few-body sector is a special case of the model considered here: the one-dimensional three-body problem with contact interactions. In Jacobi coordinates, a system of two bosons and one impurity is mapped onto a single particle moving in the plane and subject to three $\delta$-potential barriers intersecting at a common point, with the angles between the barriers determined by the mass ratio \cite{guijarro2018,kartavtsev2009,nishida2018,mcguire1964}. Zero-range few-body problems of this class remain under active study, both in their dependence on the spatial dimension \cite{rosa2022Ddim} and in their mass-imbalanced Efimov spectra \cite{oi2024efimov,blume2014heavy}; for reviews of universal few-body physics and of impurities in bosonic gases we refer to Refs.~\cite{greene2017universal,naidon2017efimov,grusdt2025impurities}. We previously studied this problem in Ref.~\cite{liu2021}, but only its bound-state sector: the spectrum was obtained from the Skornyakov--Ter-Martirosyan equations \cite{skorniakov1957}, and related impurity bound states have since been obtained by complementary methods \cite{shi2016boundstates}, while the scattering problem was left unresolved.

The present model also belongs to the broader class of ``leaky'' quantum structures, namely quantum systems with $\delta$-interactions supported on lines or curves \cite{exner2003approx,exner2008review}. Previous studies of such structures have been predominantly mathematical, focusing mainly on their spectra, resonances, and self-adjoint realizations \cite{kostrykin1999kirchhoff,exner2008review}. Closely related zero-range models have been revisited recently, from the renormalization of point scatterers in two and three dimensions \cite{loran2022multidelta} to analytically solvable Kronig--Penney chains \cite{sroczynska2025kronig}. The $\delta$ potential itself remains a standard exactly solvable model, whose spectral and scattering properties continue to be refined \cite{belloni2014delta,camblong2019levinson}. By contrast, the scattering of plane waves by intersecting repulsive $\delta$-potential barriers with $c>0$ has remained largely unexplored \cite{exner2008review}.

In this work, we solve the scattering problem for the intersecting $\delta$-potential barriers shown in Fig.~\ref{fig:schematic}. We develop an exact and systematic method in which a moving rule reduces the singular part of the projected Lippmann--Schwinger equations to a finite linear system, while the remaining nonsingular part obeys a regular integral equation. The solutions single out the angles $\theta=\pi/N$ at which the intersecting mirrors generate the dihedral group $D_N$ \cite{qiu2026}. At these angles, the number of outgoing channels remains fixed as the initial condition is varied; at generic angles, by contrast, channels open and close as the classical fold points are crossed. This classical sensitivity does not manifest itself as singular behavior in the scattering probabilities: new channels open continuously from zero probability, with finite Fisher information. Instead, it is encoded in the scattering phase shifts, thereby recovering the quantum--classical correspondence at the level of phase information. As a direct application, we solve the quantum Galperin billiard exactly: the method of images converts the walls into a reflection antisymmetry of the wave function, which collapses the projected equations to a single linear system and organizes its poles into chains, so that both the outgoing distribution and the scattering phase shifts of the classical channels follow in closed form---the quantum counterpart to the classical Galperin problem \cite{galperin2003} left open in Ref.~\cite{cai2022}.

The paper is organized as follows. In Sec.~\ref{sec:model} we define the model of intersecting $\delta$-potential barriers and derive the projected integral equations from the Lippmann--Schwinger equation. Section~\ref{sec:method} develops the exact method: a moving rule reduces the singular part of the projected equations to a finite linear system, and a regular integral equation determines the remainder. Section~\ref{sec:outgoing} constructs the atomic outgoing distribution within rigorous formal scattering theory, and closes the method with two explicit solutions, a two-barrier configuration and the three-barrier kaleidoscope. Section~\ref{sec:sensitivity} studies the initial-condition sensitivity of the outgoing distribution: at the level of the probabilities it disappears, and it is transferred to the scattering phases. Section~\ref{sec:outlook} applies the framework to the quantum Galperin billiards, where the method of images turns the walls into a reflection antisymmetry and collapses the coupled equations to a single linear system, so that the outgoing distribution and the scattering phase shifts of the classical channels follow in closed form. Section~\ref{sec:conclusion} concludes.

\section{Model and the projected integral equation}
\label{sec:model}

Consider $M$ straight $\delta$-potential barriers crossing at the origin in the plane, as shown in Fig.~\ref{fig:schematic}(a). The $m$-th barrier is characterized by its orientation angle $\theta_m \in [0, \pi)$ and its strength $c_m>0$; its unit-strength potential is supported on the line $y\cos\theta_m-x\sin\theta_m=0$ with the form 
\begin{align}
\label{eq:Vm}
V_m(x,y) := \delta\bigl(y\cos\theta_m - x\sin\theta_m\bigr),
\end{align}
and the full scattering potential is the weighted superposition
\begin{align}
\label{eq:V}
V(x,y) := \sum_{m=1}^{M} c_m V_m(x,y).
\end{align}
Since a straight line does not distinguish $\theta_m$ from $\theta_m+\pi$, the barrier angles are defined modulo $\pi$.

We work in natural units $\hbar=2m=1$, in which the Hamiltonian $H=H_0+V$ has the free part $H_0=k_x^2+k_y^2$ in momentum space. The stationary scattering state of energy $E=a^2$ ($a>0$) satisfies the Lippmann--Schwinger equation \cite{taylor1972,newton1982}
\begin{align}
\label{eq:LS}
\ket{\Psi^{(+)}}
&= \ket{\Phi_0}
+\frac{1}{a^2+\iz-k_x^2-k_y^2}\,\hat V\,\ket{\Psi^{(+)}},
\end{align}
with the outgoing ($+$) boundary condition.

In momentum space, the potential operator can be factorized into rotations and a line-integration operator. We define the rotation $R_\theta$ acting on points in $\mathbb{R}^2$ as
\begin{align}
\label{eq:Rtheta}
R_\theta(x,y) := (x\cos\theta - y\sin\theta,\; x\sin\theta + y\cos\theta),
\end{align}
so its pullback acts on functions as
\begin{align}
\label{eq:Rop}
(\hat R_\theta\psi)(x,y) := \psi\bigl(R_{-\theta}(x,y)\bigr).
\end{align}
With the vertical projection operator $\hat{I}_y$ defined by $(\hat I_y\psi)(x) := \int_{-\infty}^{\infty}\dd y\,\psi(x,y)$, each barrier factorizes as $\hat V_m = \frac{1}{2\pi}\hat R_{\theta_m}\hat I_y\hat R_{-\theta_m}$, and the full potential reads
\begin{align}
\label{eq:Vrot}
\hat V = \sum_m \frac{c_m}{2\pi}\hat R_{\theta_m}\hat I_y\hat R_{-\theta_m}.
\end{align}
Projecting the Lippmann--Schwinger equation~\eqref{eq:LS} with $\hat I_y\hat R_{-\theta_k}$, i.e.\ onto the direction of the $k$-th barrier, and defining the projected wave functions $f_k:=\hat I_y\hat R_{-\theta_k}\ket{\Psi^{(+)}}$ and $f_k^{(0)}:=\hat I_y\hat R_{-\theta_k}\ket{\Phi_0}$, we obtain the coupled integral equations
\begin{align}
\label{eq:integral}
f_k
&= f_k^{(0)}+\frac{1}{2\pi}\sum_{m=1}^{M} c_m
\hat I_y\hat G_0\hat R_{(\theta_m-\theta_k)} f_m,
\end{align}
where $\hat G_0 := (a^2+\iz-H_0)^{-1}$ is the free Green operator, which is multiplication by $1/(a^2+\iz-k_x^2-k_y^2)$ in momentum space and commutes with rotations. The kernel of each term admits the explicit representation
\begin{align}
\label{eq:kernel}
&(\hat I_y\hat G_0\hat R_\theta f)(x)
= -\frac{1}{2\sqrt{a^2-x^2}}\int_{-\infty}^{\infty}\dd\eta\,f(\eta)\notag\\
&\quad \times\Bigl[\frac{1}{\eta-\Gamma^+_\theta(x)-\iz} -\frac{1}{\eta-\Gamma^-_\theta(x)+\iz}\Bigr],
\end{align}
with
\begin{align}
\label{eq:Gamma}
\Gamma^\pm_\theta(x)
= x\cos\theta \pm \sqrt{a^2-x^2}\,\lvert\sin\theta\rvert.
\end{align}
Here and below, the square root is taken on its principal branch, with positive imaginary part for $|x|>a$. For $|x|<a$, the two quantities $\Gamma^\pm_\theta(x)$ are real, and the $\iz$ prescriptions place the corresponding poles on opposite sides of the real $\eta$ axis. The integral equations therefore contain Cauchy-type singularities associated with the positive-energy shell. For $|x|>a$, the poles move away from the integration contour. The apparent square-root singularity at $x=\pm a$ is removable, because the difference in brackets in Eq.~\eqref{eq:kernel} vanishes at the same rate as $\sqrt{a^2-x^2}$.

The corresponding negative-energy problem is considerably simpler. Setting $E=-\varkappa^2<0$ removes the on-shell singularity of the free Green operator, and the kernel becomes nonsingular on the real integration contour. In the absence of an incident wave, the driving terms also vanish, $f_k^{(0)}=0$, leaving a homogeneous system whose nontrivial solutions determine the discrete bound-state spectrum. This is the Skornyakov--Ter-Martirosyan formulation used in our earlier study of the three-body bound-state problem \cite{liu2021,skorniakov1957}.

For scattering at positive energy, by contrast, the singularity of the kernel is accompanied by a distributional driving term. For an incident plane wave of definite momentum, the free state in momentum space is
\begin{align}
\langle k_x,k_y\vert\Phi_0\rangle
= \delta(k_x-a\cos\phi_0)\,
  \delta(k_y-a\sin\phi_0),
\end{align}
and hence
\begin{align}
\label{eq:f0}
f_k^{(0)}(x)
= \delta\bigl(x-a\cos(\phi_0-\theta_k)\bigr).
\end{align}
The positive-energy equations must therefore be understood in the distributional sense. Starting from the $\delta$-function driving terms, the on-shell kernel propagates singular contributions among the coupled projected functions. Our central task is to determine this singular structure and solve the resulting equations for each incident direction $\phi_0$. Once the projected functions $f_m$ are obtained, the full scattering state follows directly from the Lippmann--Schwinger equation,
\begin{align}
\label{eq:psi}
\ket{\Psi^{(+)}}
= \ket{\Phi_0}
+\frac{1}{a^2+\iz-k_x^2-k_y^2}
\sum_m \frac{c_m}{2\pi}\,\hat R_{\theta_m} f_m.
\end{align}

As a preliminary simplification, we exploit a scaling law that maps the scattering problem at arbitrary positive energy to an equivalent one at unit energy. Denoting the set of interaction strengths by $\boldsymbol{c}=(c_1,\ldots,c_M)$, the projected functions transform as
\begin{align}
\label{eq:scaling}
f_k\bigl(x\,\big|\,\boldsymbol{c},a\bigr)
&= \frac{1}{a}\,
f_k\!\left(\frac{x}{a}\,\Big|\,\frac{\boldsymbol{c}}{a},1\right),
\end{align}
while the driving terms obey
\begin{align}
f_k^{(0)}(x\,|\,a)
=\frac{1}{a}f_k^{(0)}\!\left(\frac{x}{a}\,\Big|\,1\right),
\end{align}
by the homogeneity of the $\delta$ function in Eq.~\eqref{eq:f0}. Thus, it suffices to solve the scattering problem at $a=1$, which we assume henceforth while suppressing the dependence on $\boldsymbol{c}$ and $a$.

\begin{figure}[t]
\centering
\includegraphics[width=\columnwidth]{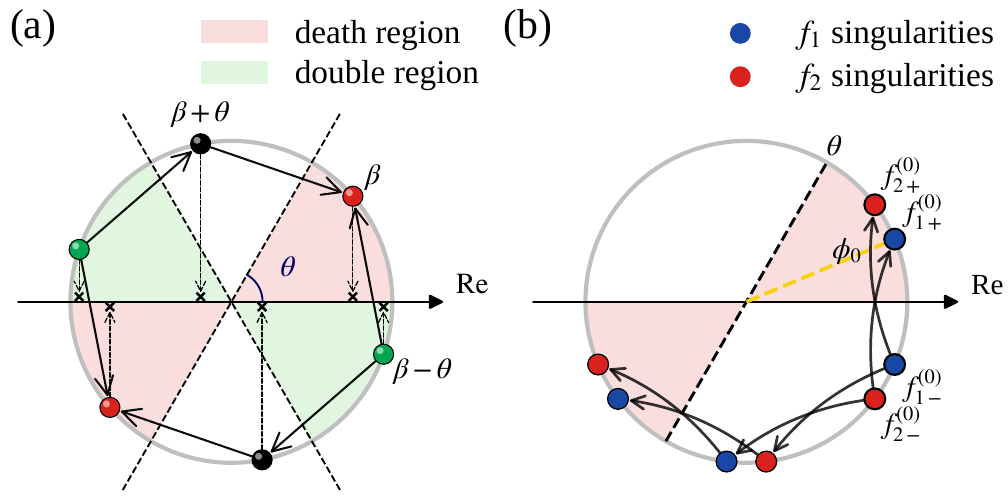}
\caption{
Schematic of the moving rule and the resulting pole closure.
(a) Unit-circle representation of a pole $\ket{\beta}$: its vertical
projection gives the pole position $\cos\beta$ on the real axis, while the
upper or lower semicircle specifies the side of the $i0^+$ prescription.
The singular action of $\hat K_\theta$ generates the allowed moves
$\beta\mapsto\beta\pm\theta$, indicated by the directed arrows. No pole is
generated in the death region (red), whereas both moves occur in the double
region (green).
(b) Example of the closure construction for two intersecting
$\delta$-potential barriers with $\theta_1=0$, $\theta_2=\theta=\pi/3$, and
incident angle $\phi_0=0.4$. Blue and red dots represent the pole sets
$C_1$ and $C_2$ of $f_1$ and $f_2$, respectively; the labeled dots are the
seed poles supplied by the incident wave. The arrows trace successive
applications of the moving rule until the finite closure is reached.
}
\label{fig:moving-rule}
\end{figure}

\section{The moving-rule method}
\label{sec:method}

At positive energy, both the kernel of Eq.~\eqref{eq:integral} and the projected wave functions must be treated distributionally. The $\delta$-function source introduces singularities on the energy shell, while the on-shell kernel propagates them among the coupled equations. A direct discretization would obscure this singular structure and is generally ill suited to a stable numerical treatment. We therefore isolate the singularities analytically and treat their positions and coefficients as algebraic data. The analysis proceeds in three steps: we introduce the elementary pole singularities, identify those injected by the incident wave, and determine how the kernel propagates them.

The elementary singularity is a simple pole whose real part lies in the energy shell $[-1,1]$, and which approaches the real axis from either above or below. Such a pole can be labeled by an angle $\beta$ and written as
\begin{align}
\label{eq:pole}
\ket{\beta}
:=  
\frac{1}{
\zeta-\cos\beta-\ii0^+\cdot\operatorname{sign}(\sin\beta)
},
\end{align}
where $\zeta$ denotes the complexified projected momentum. The angle $\beta$ encodes both the position $\cos\beta$ on the real axis and the side from which the pole approaches it [Fig.~\ref{fig:moving-rule}(a)]. In particular, $\beta$ and $2\pi-\beta$ correspond to the same real position but to opposite sides of the axis, and hence represent distinct pole distributions. The limiting cases $\sin\beta=0$ are understood through one-sided limits.

We next identify the singularities injected by the incident wave. At unit energy, the driving term in Eq.~\eqref{eq:f0} is
\begin{align}
f_k^{(0)}(x)
=
\delta\bigl(x-\cos(\phi_0-\theta_k)\bigr)
\end{align}
with $\theta_{0k} := \vert \phi_0 - \theta_k \vert$. 
The Sokhotski--Plemelj identity then gives
\begin{align}
\label{eq:seed}
f_k^{(0)}
=
\frac{1}{2\pi\ii}
\left(
\ket{\theta_{0k}}-\ket{-\theta_{0k}}
\right),
\end{align}
where $\ket{\theta_{0k}}$ and $\ket{-\theta_{0k}}$ approach the real axis from opposite sides. Thus, the incident wave supplies each projected equation with a pair of seed poles at the same on-shell momentum, with coefficients $\pm1/(2\pi\ii)$.

The remaining question is how these seed poles propagate under the kernel. The moving rule, stated in Eq.~\eqref{eq:rules} below and derived in Appendix~\ref{app:movingrule}, shows that the singular part of the action of $\hat I_y\hat G_0\hat R_\theta$ on $\ket{\beta}$ contains at most two poles of the same form. Their angles are obtained by the moves
\begin{align}
\beta\longmapsto\beta\pm\theta,
\end{align}
with the corresponding coefficients determined by the rule. In certain angular regions, referred to below as the death regions, a move produces no pole. Crucially, the moving rule preserves both defining properties of the pole family: the real part remains in $[-1,1]$, and the pole remains infinitesimally close to the real axis.

The singular sector is therefore closed under the integral equations. The driving terms consist of the seed poles $\ket{\pm\theta_{0k}}$, and every subsequent singularity is generated by repeated applications of the moving rule. It follows inductively that all singularities of the projected functions are poles of the form in Eq.~\eqref{eq:pole}. Away from the interval $[-1, 1]$, the projected functions $f_k$ are regular on the real axis.

This closed singular structure allows us to separate each projected function into singular and nonsingular parts,
\begin{align}
\label{eq:decomp}
f_k
=
f_k^{(\mathrm S)}+f_k^{(\mathrm N)},
\end{align}
where $f_k^{(\mathrm S)}$ is a superposition of the poles generated from the seeds and $f_k^{(\mathrm N)}$ is regular on the real axis. Substituting this decomposition into Eq.~\eqref{eq:integral} and applying the moving rule fixes the pole coefficients through a finite linear system. After these singular contributions are subtracted analytically, the nonsingular parts satisfy regular integral equations. We develop this construction in the remainder of this section.

\begin{figure}[t]
\centering
\includegraphics[width=\columnwidth]{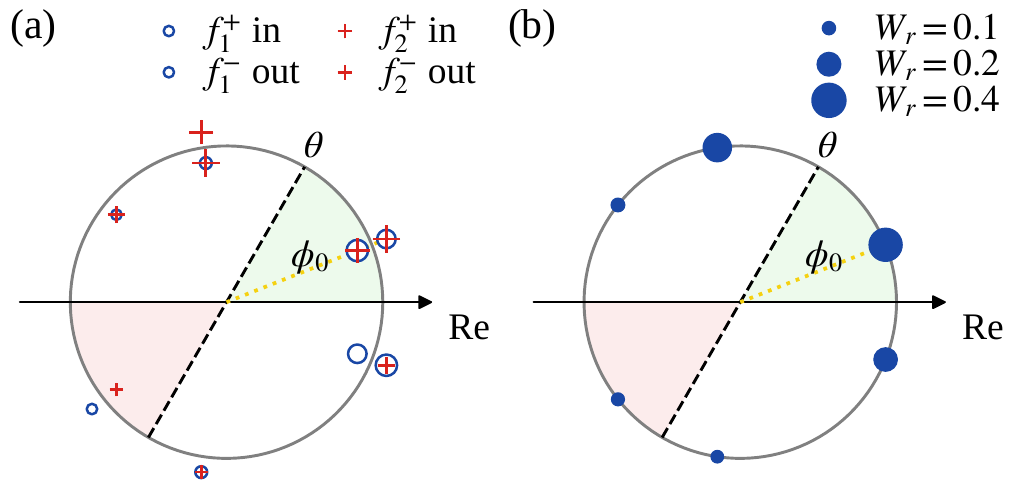}
\caption{(a) Distribution of all singular-pole triples of the singular part on
the outgoing-angle circle: blue open circles and red crosses denote triples
generated by $f_1$ and $f_2$, respectively; $\lambda=+1$ (inside the circle)
and $\lambda=-1$ (outside); the marker size encodes the amplitude modulus
$|A|$. The $\theta$ diameter and the light-green and light-red sectors mark
the barrier directions and the angular wedge. (b) Normalized atomic outgoing
distribution: one dot per outgoing channel, with size encoding the weight
$W_r$ obtained by coherent addition of same-angle triples. Parameters:
$\theta_1=0$, $\theta_2=\theta=\pi/3$, $\phi_0=0.718\,\theta/2$, $c_1=0.5$,
$c_2=1.0$; the closed form $2\pi^2\sum_r\widetilde W_r=1$ is checked numerically.}
\label{fig:triples}
\end{figure}

\subsection{The singular part}
To state the action of the kernel on a pole explicitly, we rewrite Eq.~\eqref{eq:integral} in terms of the kernel operator
\begin{align}
\hat K_\theta
:= -\frac{1}{\pi\ii}\,\hat I_y\hat G_0\hat R_\theta,
\end{align}
so that Eq.~\eqref{eq:integral} reads
\begin{align}
\label{eq:Kform}
f_k &= f_k^{(0)} - \sum_{m=1}^{M}\frac{\ii c_m}{2}\,\hat K_{\theta_{km}} f_m
\end{align}
with $\theta_{km} := \lvert\theta_k-\theta_m\rvert $ the angle between the two barrier directions. The barrier angles themselves are defined modulo $\pi$ [Eq.~\eqref{eq:V}], but the difference is taken between the representative directions used there: the kernel $\hat K_\theta$ distinguishes $\theta$ from $\pi-\theta$, so that the two barrier directions $\pm\theta$ of Sec.~\ref{sec:outlook} give $\theta_{12}=2\theta$. 
The action of the kernel on a pole is the geometric moving rule (Fig.~\ref{fig:moving-rule}):

\begin{align}
\label{eq:rules}
\hat K_\theta\ket{\beta}
& \xrightarrow{\mathrm{sing}} \frac{1}{|\sin\beta|}\Bigl[
\Theta\bigl(\sin\beta\,\sin(\beta-\theta)\bigr)\ket{\beta-\theta}
\notag\\
&\quad
-\Theta\bigl(-\sin\beta\,\sin(\beta+\theta)\bigr)\ket{\beta+\theta}\Bigr],
\end{align}
Equation~\eqref{eq:rules} is the moving rule: it states that the singular part of the action of $\hat K_\theta$ on a pole $\ket{\beta}$ consists of at most two poles of the same form, whose complete derivation from the kernel representation is given in Appendix~\ref{app:movingrule}. Here $\xrightarrow{\mathrm{sing}}$ means that only the singular part generated by the kernel action is retained, with regular terms omitted. The Heaviside function is defined by $\Theta(x)=1$ for $x>0$ and $\Theta(x)=0$ for $x<0$; its value at $x=0$ is immaterial, as the boundary cases are understood by one-sided limits. Each allowed move carries a residue factor $\pm1/|\sin\beta|$. In the death region $(0,\theta)\cup(\pi,\pi+\theta)$ no move is allowed and the propagation terminates, whereas in the double region both moves are realized [Fig.~\ref{fig:moving-rule}(a)]. Starting from the seeds $\pm\theta_{0k}$, with $\theta_{0k}=|\phi_0-\theta_k|$, for each projected function $f_k$, let $C_k$ denote the set of pole angles generated by the moving rules, initialized by the seed poles as $C_k=\{\pm\theta_{0k}\}$. Repeated application of the moving rules generates the closure
\begin{align}
\label{eq:closure}
C_k \;\leftarrow\; C_k \cup \bigcup_{m\neq k} T_{\theta_{km}}(C_m),
\end{align}
where $T_\theta$ denotes the map defined by the moving rules. The finiteness of this closure is readily seen for commensurate barrier angles and for the two-barrier problem, which are the main cases studied here. For an arbitrary barrier configuration, the result is nontrivial; we prove in Appendix~\ref{app:closure} that the moving rules nevertheless always generate a finite closure. This finiteness is essential, as it reduces the singular sector to a finite-dimensional linear system for an arbitrary barrier configuration.

On the closure, the singular part is expanded as
\begin{align}
f_k^{(\mathrm S)}(z)
= \sum_{\beta\in C_k} \frac{F_{k\beta}}{z-\cos\beta-\iz\,\operatorname{sign}(\sin\beta)}.
\end{align}
Inserting this expansion into Eq.~\eqref{eq:Kform} and matching the residues of the poles $\ket{\beta}$ with $\beta\in C_k$, the coefficients $F_{k\beta}$ are determined by a finite linear system. To state it explicitly, we number the closure of each projected function as $C_k=\{\beta_{k1},\ldots,\beta_{kN_k}\}$ and write $F_{kq}\equiv F_{k\beta_{kq}}$. The moving rule~\eqref{eq:rules} generates, for each pair $(k,m)$ of barriers, an $N_k\times N_m$ propagation matrix $K_{km}$ with entries
\begin{align}
\label{eq:Kmat}
(K_{km})_{qq'}
= \Theta\bigl(\sin\beta_{mq'}\sin(\beta_{mq'}-\theta_{km})\bigr)
\,\delta_{\beta_{kq},\beta_{mq'}-\theta_{km}}
\notag\\
\quad
- \Theta\bigl(-\sin\beta_{mq'}\sin(\beta_{mq'}+\theta_{km})\bigr)
\,\delta_{\beta_{kq},\beta_{mq'}+\theta_{km}},
\end{align}
where the angle differences are taken modulo $2\pi$ and $\delta$ is the Kronecker delta on the closure: a matrix element is nonzero only if the moved pole $\beta_{mq'}\pm\theta_{km}$ coincides with an element of $C_k$. The diagonal blocks are fixed by the identity move: $\theta_{kk}=0$ gives $K_{kk}=I_{N_k}$. It is convenient to absorb the residue factors of the moving rule into rescaled unknowns $G_{kq}:=F_{kq}/|\sin\beta_{kq}|$. With $S_k:=\mathrm{diag}(|\sin\beta_{k1}|,\ldots,|\sin\beta_{kN_k}|)$, matching residues yields the block linear system
\begin{align}
\label{eq:linear}
\sum_{m=1}^{M}\Bigl[\delta_{km}S_k + \frac{\ii c_m}{2}K_{km}\Bigr] G_m
= F_k^{(0)}
\end{align}
for $k = 1, \cdots, M$, where $G_m$ and $F_k^{(0)}$ are column vectors of length $N_m$ and $N_k$. The driving data on the right-hand side follow from the seed decomposition of Eq.~\eqref{eq:seed}: the only nonzero entries of $F_k^{(0)}$ are
\begin{align}
F_k^{(0)}(\theta_{0k}) = +\frac{1}{2\pi\ii},
\quad
F_k^{(0)}(2\pi-\theta_{0k}) = -\frac{1}{2\pi\ii},
\end{align}
i.e.\ the seed pair $\ket{\pm\theta_{0k}}$ of the incident wave, in the convention of Eq.~\eqref{eq:seed}. Solving Eq.~\eqref{eq:linear} for $G$ and restoring $F_{kq}=|\sin\beta_{kq}|\,G_{kq}$ determines the singular part completely. This is the linear system alluded to in Eq.~\eqref{eq:decomp}.

The solution of the singular linear system is the complete singular content of the projected functions. For the two-barrier configuration with $\theta=\pi/3$ and unequal strengths $c_1=0.5$, $c_2=1.0$ (Fig.~\ref{fig:triples}), the resulting poles of $f_1$ and $f_2$, mapped onto the circle of outgoing angles and weighted by their residues, are displayed in Fig.~\ref{fig:triples}(a), together with the atomic outgoing distribution they generate in Fig.~\ref{fig:triples}(b); the way this distribution is computed from these poles is given in Sec.~\ref{sec:outgoing}.

As another instance we solve the three-barrier kaleidoscope, with barriers at $0$, $\pi/3$ and $2\pi/3$ of equal strength $c=1$ and an incident plane wave at $22^\circ$ inside the fundamental domain. Here the symmetry alone fixes the result: the three barriers form the mirror system of the dihedral group $D_3$, so the poles generated by the moving rule are distributed uniformly around the circle, eight on each barrier, and the outgoing atoms fall into the six directions of the corresponding mirror-image set. The distribution is shown in Fig.~\ref{fig:kaleidoscope}.

\begin{figure}[t]
\centering
\includegraphics[width=\columnwidth]{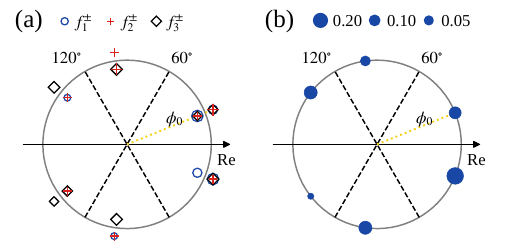}
\caption{Three-barrier kaleidoscope: (a) the poles of the scattering state, colored by the barrier that generates them, and (b) the resulting outgoing channels.}
\label{fig:kaleidoscope}
\end{figure}

\subsection{The non-singular part}

\begin{figure}[t]
\centering
\includegraphics[width=\columnwidth]{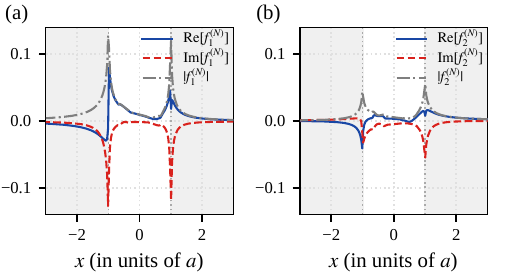}
\caption{Numerical solution of the non-singular part for the two-barrier
configuration of Fig.~\ref{fig:moving-rule}: $\theta_1=0$,
$\theta_2=\theta=\pi/3$, incident angle $\phi_0=0.718\,\theta/2$,
$c_1=0.5$, $c_2=1.0$. Left and right panels show $f_1^{(\mathrm N)}(x)$ and
$f_2^{(\mathrm N)}(x)$: solid blue, dashed red, and dash-dotted gray lines
denote the real part, the imaginary part, and the modulus, respectively.
Vertical dotted lines mark $x=\pm1$, the branch points of
$s(x)=\sqrt{1-x^2}$; the light-gray background marks the region
$|x|>1$ outside the energy shell. Grid: $N_x=400$ nodes on
$[-K_{\mathrm{cut}},K_{\mathrm{cut}}]=[-6,6]$, boundary-value regulator
$\varepsilon=10^{-8}$; the linear system is solved to residual
$\sim10^{-16}$ and agrees with fixed-point iteration to $\sim10^{-9}$.}
\label{fig:normal}
\end{figure}
With the singular part at hand, it remains to determine the normal part $f_k^{(\mathrm N)}$. A direct evaluation of $\hat K_\theta f_m^{(\mathrm S)}$ on the real axis is not feasible: although the pole positions and residues of $f_m^{(\mathrm S)}$ are already known, the full action $\hat K_\theta$ on them still contains poles infinitesimally close to the real axis, and subtracting the singular pieces numerically would suffer from catastrophic cancellation near these poles. The poles must therefore be removed analytically before any grid evaluation. To this end, we define the operator $\mathcal S_\theta$ so that $\mathcal S_\theta\ket{\beta}$ is exactly the right-hand side of the moving rule, Eq.~\eqref{eq:rules}: the newly generated poles $\ket{\beta\pm\theta}$ together with their residue factors. The regularized remainder of a single pole is then defined by
\begin{align}
\label{eq:Qdef}
Q_{\theta,\beta}(z)
:= s(z)\left(\hat K_\theta\ket{\beta}-\mathcal S_\theta\ket{\beta}\right),
\end{align}
where the factor $s(z):=\sqrt{1-z^2}$ is included because the same factor multiplies the non-singular equations below. The subtraction in Eq.~\eqref{eq:Qdef} is exact at the level of residues: the two terms have identical pole structure, so all poles of $\hat K_\theta\ket{\beta}$ that belong to the singular sector are cancelled analytically, and $Q_{\theta,\beta}(z)$ is a regular function of $z$ on the real axis. Carrying it out sector by sector in $\beta$, one obtains closed-form expressions for $Q_{\theta,\beta}$ that are valid for every barrier angle $\theta\in(0,\pi)$; the complete piecewise forms are collected in Appendix~\ref{app:Q}. The sectors follow the output pattern of the moving rule~Eq.~\eqref{eq:rules} for the given pole: when no move is allowed (death region), the remainder reduces to the full action, with no pole to subtract; when one or both moves are allowed, the newly generated poles are subtracted together with their exact residues. In every sector, $Q_{\theta,\beta}(z)$ is therefore a regular function of $z$ on the real axis, with sector boundaries understood by one-sided limits and the apparent endpoint singularities at $x=\pm1$ being the integrable branch points of $s(z)=\sqrt{1-z^2}$, not poles. These forms retain the finite contribution of the singular scattering to the non-singular part and are evaluated as ordinary functions on the momentum grid.

The source terms of the non-singular equations are generated by these remainders. Substituting the expansion of $f_m^{(\mathrm S)}$ into Eq.~\eqref{eq:Kform} and collecting everything that is not absorbed by the singular linear system, one finds that $h_k$ is built from the known residues $F_{m\beta'}$ as
\begin{align}
\label{eq:normal-source}
h_k(z)
= -\sum_{m=1}^{M}\sum_{\beta'\in C_m}
\frac{\ii c_m}{2}\,F_{m\beta'}\,
Q_{\theta_{km},\beta'}(z),
\end{align}
which is not the original incident driving term but the finite remainder left by the singular solution after all kernel actions. For the regular unknowns $f_k^{(\mathrm N)}$, the self-action satisfies the identity $s\hat K_0 f_k^{(\mathrm N)}=f_k^{(\mathrm N)}$, so it can be moved to the left-hand side; defining $D_k(z):=s(z)+\ii c_k/2$ and writing $K_p^{(\theta)}:=s\hat K_\theta$ for the regularized kernel, the normal parts obey the coupled integral equations
\begin{align}
\label{eq:normal-system}
D_k f_k^{(\mathrm N)}
+ \sum_{m\neq k}\frac{\ii c_m}{2}K_p^{(\theta_{km})} f_m^{(\mathrm N)}
= h_k
\end{align}
for $k = 1, \cdots, M$,  where the boundary-value operator acts as
\begin{align}
\label{eq:Kp}
(K_p^{(\theta)}g)(x)
&= \frac{1}{2\pi\ii}\left[
\int_{-\infty}^{\infty}\dd y\,
\frac{g(y)}{y-r_+(x)-\iz}
\right.\notag\\
&\quad\quad\left.
- \int_{-\infty}^{\infty}\dd y\,
\frac{g(y)}{y-r_-(x)+\iz}
\right],
\end{align}
with $r_\pm(x):=x\cos\theta\pm s(x)\sin\theta$. For $x\in[-1,1]$ the poles $r_\pm(x)$ lie on the real axis, and the $\mp\iz$ prescriptions define the boundary values from the two sides; the operator therefore contains principal values and jump terms, and must be handled as a Cauchy-type integral even though $f_k^{(\mathrm N)}$ itself is regular.

Equation~\eqref{eq:normal-system} is a regular Fredholm system of the second kind on the real axis, which is solved numerically on a uniform momentum grid. The integrals of Eq.~\eqref{eq:Kp} are evaluated by piecewise-linear interpolation on the grid, with each segment integrated analytically; outside the grid, the solution decays as $1/y^2$, and the tails are integrated analytically as well. The resulting linear system is solved directly and cross-checked by fixed-point iteration, which agree to numerical precision. Figure~\ref{fig:normal} shows the resulting normal parts for the two-barrier configuration of Fig.~\ref{fig:moving-rule}; they are smooth on $(-1,1)$ and develop the integrable endpoint structure $s(x)\propto\sqrt{1-x^2}$ at $x=\pm1$, which is not a pole to be subtracted but the branch-point singularity of the kernel, harmless for the integral equations. The full solution is assembled as $f_k=f_k^{(\mathrm S)}+f_k^{(\mathrm N)}$, with the singular part given analytically by its poles and residues and the normal part stored on the grid.

\section{Atomic outgoing distribution}

\label{sec:outgoing}

To assemble scattering data from the projected functions, we return to the Lippmann--Schwinger formalism of Sec.~\ref{sec:model}. In standard scattering theory the transition operator is defined by $\ket{\Psi^{(+)}}=\ket{\Phi_0}+\hat G_0^{(+)}(E)\hat T(E)\ket{\Phi_0}$, equivalently $\hat T(E)=\hat V+\hat V\hat G_0^{(+)}(E)\hat T(E)$ \cite{taylor1972,newton1982}. Comparing this with the explicit solution of Eq.~\eqref{eq:psi}, in which the action of the potential is resolved into the rotated projected functions, we read off the momentum-space matrix elements of the transition operator,
\begin{align}
\label{eq:Tmatrix}
\bra{\boldsymbol{k}}\hat T(a^2+\iz)\ket{\boldsymbol{k}_0}
= \frac{1}{2\pi}\sum_{m=1}^{M} c_m\, f_m(\boldsymbol{k}\cdot\boldsymbol{t}_m),
\end{align}
with $\boldsymbol{t}_m:=(\cos\theta_m,\sin\theta_m)$ the unit vector along the $m$-th barrier: the $m$-th barrier scatters the wave with a weight that depends on the outgoing momentum only through its component along the barrier. The projected functions $f_m$ are therefore intermediate quantities, out of which the transition amplitudes are assembled, and not scattering amplitudes themselves.

The scattering is elastic, so the incident and final momenta share the magnitude $a$; by the scaling law of Sec.~\ref{sec:model} it suffices to evaluate the projected functions at unit energy, but we keep $a$ explicit in the reduction. Writing $\boldsymbol{k}_i=a(\cos\phi_0,\sin\phi_0)$ and $\boldsymbol{k}_f=a(\cos\varphi,\sin\varphi)$, with $\varphi$ the outgoing angle, the contraction with the barrier direction gives $\boldsymbol{k}_f\cdot\boldsymbol{t}_m=a\cos(\varphi-\theta_m)$, and the restriction of Eq.~\eqref{eq:Tmatrix} to the energy shell defines the angular transition amplitude
\begin{align}
\label{eq:Tshell}
T(\varphi,\phi_0)
:= \frac{1}{2\pi}\sum_{m=1}^{M} c_m\,
f_m\!\bigl(a\cos(\varphi-\theta_m)\bigr).
\end{align}
The contributions of the different barriers are added coherently at the level of amplitudes, before any modulus is taken.

The on-shell $S$ matrix is obtained from the standard reduction formula \cite{taylor1972,newton1982}
\begin{align}
\label{eq:Sfull}
\bra{\boldsymbol{k}_f}\hat S\ket{\boldsymbol{k}_i}
& = \delta^{(2)}(\boldsymbol{k}_f-\boldsymbol{k}_i)\notag\\
& - 2\pi\ii\,\delta(k_f^2-k_i^2)\,
\bra{\boldsymbol{k}_f}\hat T(a^2+\iz)\ket{\boldsymbol{k}_i}.
\end{align}
In the plane-wave normalization of Sec.~\ref{sec:model}, $\delta^{(2)}(\boldsymbol{k}_f-\boldsymbol{k}_i)=\frac{\delta(k_f-a)}{a}\,\delta_{2\pi}(\varphi-\phi_0)$ and $\delta(k_f^2-a^2)=\frac{\delta(k_f-a)}{2a}$, where $\delta_{2\pi}$ denotes the $2\pi$-periodic delta function on the circle of angles. Extracting the common radial factor $\delta(k_f-a)/a$, the angular $S$ matrix reads
\begin{align}
\label{eq:Sangular}
S(\varphi,\phi_0)
= \delta_{2\pi}(\varphi-\phi_0) - \ii\pi\, T(\varphi,\phi_0),
\end{align}
the first term describing the unscattered forward wave and the second the reflected and deflected channels generated by the barriers.

An infinitely long straight barrier scatters into sharply defined directions, so the angular $S$ matrix is a distribution rather than a function: it contains $\delta$ functions, and its modulus squared is not defined in the usual sense. The standard resolution is to work with finite-width wave packets \cite{newton1982}. For the discrete singular channels of the present model we use an equivalent regularization that is better adapted to the pole structure obtained in Sec.~\ref{sec:method}: each angular singularity is represented by a pole triple
\begin{align}
\label{eq:triple}
\mathfrak p := (A,\varphi_0,\lambda),\quad
\mathcal C_\varepsilon[\mathfrak p](\varphi)
:= \frac{A}{\varphi-\varphi_0-\ii\varepsilon\lambda},
\end{align}
where $A\in\mathbb{C}$ is the amplitude, $\varphi_0\in[0,2\pi)$ the position on the circle of angles, and $\lambda=\pm1$ records the side of the real axis from which the pole is approached: the pole of $\mathcal C_\varepsilon[\mathfrak p]$ lies at $\varphi_0+\ii\varepsilon\lambda$, above the axis for $\lambda=+1$ and below for $\lambda=-1$. Since $\varepsilon\lambda\to0$, only the side has observable content, and the normalization of $\lambda$ to $\pm1$ is a convention; angles are understood modulo $2\pi$. The unit term of Eq.~\eqref{eq:Sangular} admits the Sokhotski--Plemelj decomposition
\begin{align}
\label{eq:delta-pair}
\delta_{2\pi}(\varphi-\phi_0)
&= \lim_{\varepsilon\to0^+}\frac{1}{2\pi\ii}
\Bigl[\frac{1}{\varphi-\phi_0-\ii\varepsilon}
-\frac{1}{\varphi-\phi_0+\ii\varepsilon}\Bigr],
\end{align}
so the incident $\delta$ peak is the pair of triples
\begin{align}
\label{eq:incident}
\mathfrak d_+
:= \Bigl(\frac{1}{2\pi\ii},\ \phi_0,\ +1\Bigr),\quad
\mathfrak d_-
:= \Bigl(-\frac{1}{2\pi\ii},\ \phi_0,\ -1\Bigr),
\end{align}
which share the angle $\phi_0$ but approach it from opposite sides.

Each pole of the singular part generates two outgoing triples. The singular expansion of Sec.~\ref{sec:method} is obtained at unit energy, so we set $a=1$ when substituting it into $-\ii\pi T(\varphi,\phi_0)$; changing variables to $x=\cos(\varphi-\theta_m)$, a pole $F_{mq}\ket{\beta_{mq}}$ produces two triples, located at the two roots $\varphi=\theta_m\pm\beta_{mq}$ of $x=\cos\beta_{mq}$:
\begin{align}
\label{eq:scat-triples}
\mathfrak s_{mq}^{(-)}
&= \Bigl(\frac{\ii c_m F_{mq}}{2\sin\beta_{mq}},\;
[\theta_m+\beta_{mq}]_{2\pi},\; -1\Bigr),\\
\mathfrak s_{mq}^{(+)}
&= \Bigl(-\frac{\ii c_m F_{mq}}{2\sin\beta_{mq}},\;
[\theta_m-\beta_{mq}]_{2\pi},\; +1\Bigr),\label{eq:triple-plus}
\end{align}
where $[\cdot]_{2\pi}$ denotes reduction modulo $2\pi$ into $[0,2\pi)$. The coefficient $-1/2$ originates from $-\ii\pi T=-\tfrac{\ii}{2}\sum_m c_m f_m$, and the factor $1/\sin\beta_{mq}$ is the Jacobian of the change of variables at the two roots. The complete set of discrete singularities of the angular $S$ matrix is therefore the finite collection $\mathcal Q_S:=\{\mathfrak d_+,\mathfrak d_-\}\cup\{\mathfrak s_{mq}^{(\pm)}\}_{m,q}$.

The continuous part of the $S$ matrix, generated by the normal parts through $S^{(N)}(\varphi,\phi_0):=-\tfrac{\ii}{2}\sum_m c_m f_m^{(\mathrm N)}\!\bigl(\cos(\varphi-\theta_m)\bigr)$, carries no discrete weight in the limit below. To extract the weights of the singular channels, regularize the angular measure by
\begin{align}
\label{eq:measure}
\dd\mu_\varepsilon(\varphi)
:= \frac{\varepsilon}{\pi}\,
\Bigl|S_\varepsilon^{(S)}(\varphi,\phi_0)+S^{(N)}(\varphi,\phi_0)\Bigr|^2\dd\varphi,
\end{align}
with $S_\varepsilon^{(S)}:=\sum_{\mathfrak p\in\mathcal Q_S}\mathcal C_\varepsilon[\mathfrak p]$ the regularized singular part. The prefactor $\varepsilon/\pi$ compensates the $1/\varepsilon$ divergence of $\lvert\mathcal C_\varepsilon[\mathfrak p]\rvert^2$, leaving each isolated pole with the finite weight $|A|^2$. The normal part and its cross terms with the singular part contribute $O(\varepsilon)$ and $O(\varepsilon\lvert\log\varepsilon\rvert)$, respectively, and vanish in the limit $\varepsilon\to0^+$; the atomic weights are therefore independent of the normal parts, which enter only the continuous part of the distribution.

Two triples $\mathfrak p_i=(A_i,\varphi_i,\lambda_i)$ and $\mathfrak p_j=(A_j,\varphi_j,\lambda_j)$ contribute a finite overlap only when they occupy the same angle and approach it from the same side,
\begin{align}
\label{eq:overlap}
\lim_{\varepsilon\to0^+}\frac{\varepsilon}{\pi}\int_0^{2\pi}\dd\varphi\,
&\mathcal C_\varepsilon[\mathfrak p_i](\varphi)\,
\overline{\mathcal C_\varepsilon[\mathfrak p_j](\varphi)}\notag\\
&= A_iA_j^*\;\delta_{\varphi_i,\varphi_j}\,\delta_{\lambda_i,\lambda_j},
\end{align}
the same-side condition keeping the two regularized poles on the same side of the axis so that the overlapping Lorentzians leave the finite product of amplitudes. Grouping $\mathcal Q_S$ by common angle $\varphi_r$ and by side, with $I_r^\pm:=\{\mathfrak p\in\mathcal Q_S:\varphi_0=\varphi_r,\ \lambda=\pm1\}$ the triples at $\varphi_r$ approaching from the side $\pm1$, the weak limit of the measure in Eq.~\eqref{eq:measure} is an atomic measure with unnormalized weights
\begin{align}
\label{eq:Wr}
\widetilde W_r
:= \Bigl\lvert\sum_{I_r^+}A\Bigr\rvert^2
+\Bigl\lvert\sum_{I_r^-}A\Bigr\rvert^2,
\end{align}
the coherent addition within each side preserving the interference between different scattering paths that end at the same angle from the same side. Unitarity fixes the overall scale of these weights through the incident probability. With $\kappa_j$ the coefficient carried by the $j$-th incident component, $2\pi^2\sum_r\widetilde W_r=\sum_j\lvert\kappa_j\rvert^2$, the sum running over the incident components. For the single unit-amplitude plane wave of this section $\sum_j\lvert\kappa_j\rvert^2=1$, so that $\sum_r\widetilde W_r=1/(2\pi^2)$ and the normalized probabilities $W_r:=2\pi^2\widetilde W_r$ sum to one, $\sum_r W_r=1$. A multi-component incident state is treated in the same way, its total weight $\sum_j\lvert\kappa_j\rvert^2$ replacing the unit scale. The outgoing distribution is the probability measure
\begin{align}
\label{eq:atomic}
\rho(\varphi)
= \sum_r W_r\,\delta(\varphi-\varphi_r),
\end{align}

The overall scale is fixed by the incident pair: in the absence of scattering the distribution consists of the single peak at $\phi_0$, whose unnormalized weight is $\widetilde W_{\mathrm{in}}=1/(4\pi^2)+1/(4\pi^2)=1/(2\pi^2)$, the sum of $|A|^2=1/(4\pi^2)$ from each of $\mathfrak d_\pm$. The normalization $W_r=2\pi^2\widetilde W_r$ therefore assigns unit forward weight to free propagation, and the closed form $2\pi^2\sum_r\widetilde W_r=1$ is satisfied by the atomic distribution generated from the singular parts. At the forward angle $\varphi_r=\phi_0$, the incident pair interferes with every scattering triple sharing its angle and side: the forward peak is the result of the complete $S$ matrix, not the sum of the incident and scattered intensities. Once the singular linear system of Sec.~\ref{sec:method} is solved, the triples of $\mathcal Q_S$ are at hand and the atomic weights of Eq.~\eqref{eq:atomic} are evaluated directly, with no further input from the normal parts; the resulting outgoing angular distribution for the two-barrier example is shown in Fig.~\ref{fig:triples}(b). This completes the calculation of the scattering model formulated in Sec.~\ref{sec:model}: the projected integral equations are resolved into a finite singular system and a regular remainder, and the outgoing distribution follows from the singular part alone.

\begin{figure}[t]
\includegraphics[width=0.9\columnwidth]{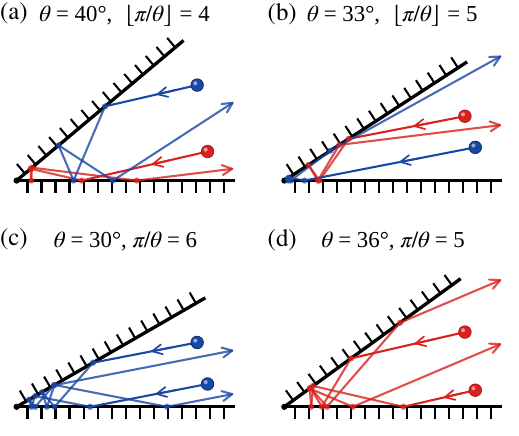}
\caption{Four typical classical scatterings.}
\label{fig:galperin-classical}
\end{figure}

\section{Sensitivity to the incident angle}
\label{sec:sensitivity}

Classical scattering often depends singularly on the initial conditions: the deflection function is stationary at rainbow angles, the intensity diverges on caustics, and the outgoing directions can jump at critical values. Where this classical singularity information goes in the quantum theory is a long-standing question: Wigner's threshold law gives a square-root cusp in the cross section when a channel opens \cite{wigner1948threshold}; the semiclassical analysis of Ford and Wheeler shows that the classical deflection function is already encoded in the quantum phase shifts and their rainbow oscillations \cite{ford1959semiclassical,ford1959application,berry1966rainbow}; and catastrophe optics replaces the classical divergence by a universal diffraction pattern, a wave crossing a caustic acquiring a Maslov phase \cite{berry1980catastrophe,marcuse1976caustic,nye1974dislocations}. The common picture is that the singularity information is not lost, but moves from the intensity distribution into phases and oscillations. These results are largely semiclassical, and their classical background is usually a (piecewise) smooth dynamics. The model solved in this work---a plane wave scattered by two intersecting $\delta$-potential barriers [Fig.~\ref{fig:schematic}(b)]---has instead a nonsmooth classical counterpart: a particle bouncing elastically in a hard wedge of opening angle $\theta$, where the first wall hit fixes the entire collision sequence and hence the escape channel, so that the channel assignment is a piecewise function of the incident angle. Its jumps at the classical fold points are the sensitivity of interest---discrete and geometric, not an exponential dynamical instability; and, the model being exactly solvable, both the classical fold structure and its fate in the quantum scattering can be followed without approximation.

Before turning to the quantum scattering problem, we first recall the classical scattering. Classical scattering falls into the four cases shown in Fig.~\ref{fig:galperin-classical}. In cases (a) and (b) fold points are present, and the outgoing angle jumps as the incident angle is varied, the jump depending on which wall is hit first; in cases (c) and (d), by contrast, the outgoing angle varies continuously with the incident angle. The blue and red trajectories indicate whether the number of collisions is even or odd, respectively.

Quantum mechanically this picture is lost: by the uncertainty principle, a wave of definite momentum is infinitely extended in the transverse direction, so which wall is struck first is not a well-posed question, and the two classical branches must be superposed coherently. Figure~\ref{fig:classical-vs-quantum} shows what this coherence produces: the single classical channel of unit weight is replaced, at every incident angle, by two channels of comparable weight lying on the two classical leaves. The scattering problem must therefore be solved exactly.

\begin{figure}[t]
\centering
\includegraphics[width=\columnwidth]{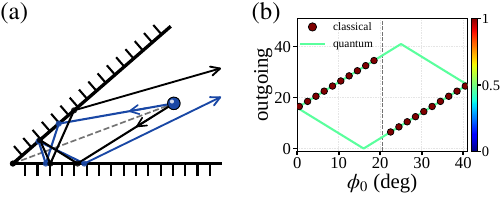}
\caption{Classical versus quantum outgoing channels for $\theta=41^\circ$: (a) the two classical trajectories; (b) the classical single channel (points, unit weight) and the two quantum channels at $c=10$ (lines, color coding $W_r$).}
\label{fig:classical-vs-quantum}
\end{figure}

\begin{figure}[t]
\centering
\includegraphics[width=\columnwidth]{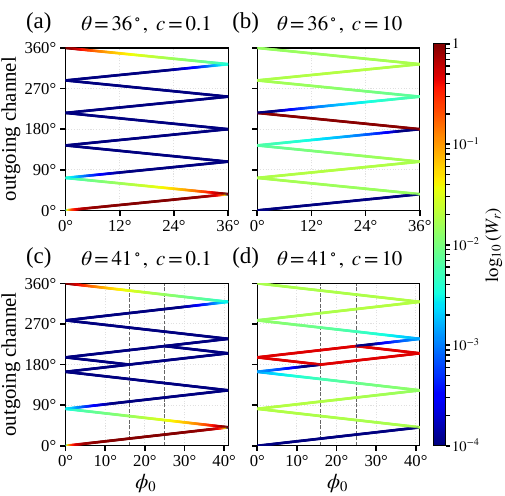}
\caption{Outgoing channels versus incident angle $\phi_0$. Each
channel is drawn as a point at its outgoing angle $\varphi$ (measured
anticlockwise, modulo $2\pi$) and colored by its weight $W_r$ of
Eq.~\eqref{eq:atomic} on a $\log_{10}$ scale, weights below $10^{-4}$
being shown in the darkest blue. Panels (a),(b) are for the rational
angle $\theta=\pi/5$ and (c),(d) for the generic angle $\theta=41^\circ$,
with (a),(c) at $c=0.1$ and (b),(d) at $c=10$. At $\pi/5$ the number of
channels stays at $10$ for all $\phi_0$; at $41^\circ$ a single pair of
them is created and destroyed at the classical fold points $\phi_0=r$ and
$\phi_0=\theta-r$, marked by the dashed lines at $16^\circ$ and
$25^\circ$. In panels (c) and (d) the new channel is seen to open from
zero weight---it starts in the darkest blue---and to grow away from the
fold, so that the probabilities remain continuous through the event even
though the channel structure changes. These positions are fixed by
geometry alone and do not depend on $c$.}
\label{fig:channel-map}
\end{figure}

\subsection{Smooth opening of the ghost channel}

The equations to be solved are those set up above: the projected functions obey Eq.~\eqref{eq:Kform}, their singular parts are generated from the seed poles of Eq.~\eqref{eq:seed} by the moving rule Eq.~\eqref{eq:rules}, and closing that rule as in Eq.~\eqref{eq:closure} reduces the singular sector to the linear system Eq.~\eqref{eq:linear}. For the two barriers of Fig.~\ref{fig:schematic}(b), whose closure is illustrated in Fig.~\ref{fig:moving-rule}(b), this system reads explicitly
\begin{align}
\Bigl(S_1+\frac{\ii c_1}{2}I\Bigr)G_1+\frac{\ii c_2}{2}K_{12}G_2&=F_1^{(0)},\notag\\[2pt]
\frac{\ii c_1}{2}K_{21}G_1+\Bigl(S_2+\frac{\ii c_2}{2}I\Bigr)G_2&=F_2^{(0)},
\label{eq:M2linear}
\end{align}
where $S_1=\operatorname{diag}(|\sin\beta_{1q}|)$ and $S_2=\operatorname{diag}(|\sin\beta_{2q}|)$ are the diagonal matrices of closure angles of the two projected functions and $G_{1q}=F_{1q}/|\sin\beta_{1q}|$, $G_{2q}=F_{2q}/|\sin\beta_{2q}|$ are the rescaled unknowns, $K_{12}$ and $K_{21}$ are the propagation matrices between the two closures, read off the moving rule with $\theta_{12}=\theta_{21}=\theta$ (entries $0$ and $\pm1$, the diagonal blocks being the identities), and the right-hand sides contain only the seed pairs: with $\theta_{01}=\phi_0$ and $\theta_{02}=\theta-\phi_0$, $F_1^{(0)}(\phi_0)=+1/(2\pi\ii)$, $F_1^{(0)}(2\pi-\phi_0)=-1/(2\pi\ii)$ and $F_2^{(0)}(\theta-\phi_0)=+1/(2\pi\ii)$, $F_2^{(0)}(2\pi-\theta+\phi_0)=-1/(2\pi\ii)$. Solving this system for $G_1$ and $G_2$ and restoring $F_{1q}=|\sin\beta_{1q}|G_{1q}$ and $F_{2q}=|\sin\beta_{2q}|G_{2q}$ determines the singular parts of both projected functions completely.

The weights plotted below follow from the coefficients determined here through the chain of Secs.~\ref{sec:method} and~\ref{sec:outgoing}: restoring $F_{kq}=|\sin\beta_{kq}|G_{kq}$ converts the solution of Eq.~\eqref{eq:linear} into the coefficients of the singular expansion, each pole then generates the two $S$-matrix triples of Eq.~\eqref{eq:scat-triples}, and the coherent sums of their amplitudes define the unnormalized weights of Eq.~\eqref{eq:Wr}, which unitarity normalizes into the outgoing probabilities of Eq.~\eqref{eq:atomic}.

Figure~\ref{fig:channel-map} displays the channel structure as a function of $\phi_0$. For the generic angle $\theta=41^\circ$ [panels~(c) and~(d)] the channels are created and destroyed as $\phi_0$ is varied, the events occurring at the classical fold points $\phi_0=r$ and $\phi_0=\theta-r$, with $r=\pi-N\theta$ and $N=\lfloor\pi/\theta\rfloor$; for the rational angle $\theta=\pi/5$ [panels~(a) and~(b)] their number is instead constant. The opening is soft: a channel is born with vanishing weight and grows with $\phi_0$, so that the probabilities stay continuous while the carrier structure changes. Here and below, the channel born at a fold---existing only on one side of it, with a weight that vanishes at the fold---is called the ghost channel. The event positions are fixed by geometry alone; the strength $c$ affects only the weights.

\begin{figure}[t]
\centering
\includegraphics[width=\columnwidth]{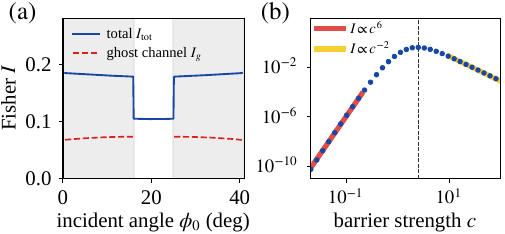}
\caption{Fisher information of the outgoing channel label for $\theta=41^\circ$: (a) total and ghost-channel values versus incident angle at $c=10$; (b) ghost-channel value versus barrier strength, with the asymptotic slopes.}
\label{fig:fisher}
\end{figure}

The soft opening has a quantitative counterpart, best expressed through the Fisher information of the distribution with respect to the incident angle. Since the outgoing state is an atomic measure and no substructure within a channel is resolved, the relevant quantity is the Fisher information of the channel label,
\begin{equation}
I(\phi_0)=\sum_r\frac{(\partial_{\phi_0}W_r)^2}{W_r},
\label{eq:fisher}
\end{equation}
the sum running over the outgoing channels. A discontinuous creation of a channel would produce a step in $I$; the soft opening instead gives $W\propto C(c)\Delta^2$ near a fold, so that the ghost channel contributes $I_{\mathrm g}=(\partial_\Delta W)^2/W=4C(c)=O(1)$, with no critical enhancement. The step that the total $I$ nevertheless shows at a fold is therefore carried entirely by the newly opened channel, the channels present on both sides contributing continuously. Figure~\ref{fig:fisher} displays both the angular dependence at $c=10$ and the dependence of the ghost-channel value on the barrier strength: at strong coupling $I\propto c^{-2}$, whereas at weak coupling the exponent is governed by the number of virtual collisions needed to build the new channel, as the explicit solution below makes clear. The curve has a maximum $I^*\simeq0.41$ near $c^*\simeq2.5$, and the fold positions themselves are fixed by geometry.

The new channel therefore leaves only a finite trace in the distribution: it contributes a finite Fisher information, and no critical enhancement accompanies its birth. If the scattering is observed through the outgoing probabilities alone, the initial-condition sensitivity of the classical problem---the discrete change of the escape channel at a fold---appears to have disappeared altogether: the distribution stays smooth, and the only remnant of the event is the finite step of Eq.~\eqref{eq:fisher} in $I$.

\begin{figure}[t]
\centering
\includegraphics[width=\columnwidth]{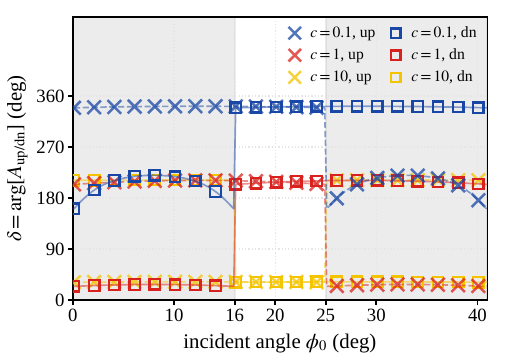}
\caption{Scattering phase shifts of the two dominant channels for $\theta=41^\circ$ at $c=0.1$, $1$ and $10$; the shaded and white intervals mark regions of different channel number, whose boundaries are the classical fold points.}
\label{fig:phase-shift}
\end{figure}

The sensitivity is not lost, however: it is transferred to the phases of the channels that continue the classical branches. Following the two dominant channels of the window $(\pi,\pi+\theta)$, which are the quantum counterparts of the two classical leaves, and referring their amplitudes to a common half-plane---multiplying by $+\ii$ for $\lambda=+1$ and $-\ii$ for $\lambda=-1$---the phase shifts so obtained are shown in Fig.~\ref{fig:phase-shift}; the two amplitudes are written $A_{\mathrm{up}}$ and $A_{\mathrm{dn}}$, the labels taken from the classical leaf that each channel continues. Each phase stays on a plateau within the interval delimited by the classical fold points and jumps by $\pi$ when a fold is crossed, the plateau value recording the collision sequence of the corresponding chain, while the fold positions are fixed by geometry and unmoved by $c$. This is the structural counterpart of a well-established theme in the theory of singularities of wave fields: the caustic at which the classical intensity diverges is crossed by a Maslov phase jump \cite{marcuse1976caustic,berry1980catastrophe}, the stationary points of the classical deflection function reappear as the oscillations of the quantum cross section \cite{ford1959semiclassical,berry1966rainbow}, and the singularities of the classical trajectories survive in the phases of the $S$-matrix elements \cite{miller1970classicalS}.

The manner of the transfer is, however, different from that of these precedents. There the classical singularity leaves a threshold structure in the distribution itself, with a form characteristic of its type: the opening of a channel produces the square-root threshold behavior $\sigma\propto(E-E_{\mathrm{th}})^{1/2}$ of Wigner's law, whose slope diverges at the threshold \cite{wigner1948threshold}, while at a fold caustic---a rainbow---the classical intensity diverges as $|\delta|^{-1/2}$ in the distance $\delta$ from the caustic and the wave field is regularized by an Airy-type pattern with its characteristic enhancement \cite{berry1966rainbow,berry1980catastrophe}. Here no such structure appears. The new channel opens smoothly, its weight vanishing quadratically,
\begin{equation}
W\propto C(c)\,\Delta^{2}\quad(\Delta\to0),
\label{eq:soft-open}
\end{equation}
so that its Fisher information stays finite and constant through the event,
\begin{equation}
I_{\mathrm g}=\frac{(\partial_\Delta W)^{2}}{W}=4C(c) \in O(1),
\label{eq:fisher-ghost}
\end{equation}
and neither a cusp nor an enhancement marks the distribution. The difference can be traced to the nature of the corresponding classical event: at a fold the outgoing direction reaches the wall direction, rather than a stationary point of the deflection function being attained, so that the classical singularity is a boundary effect and its transfer to the quantum problem is carried by the phase alone. The explicit form of that phase follows from the closed-form solution given below.

\subsection{Analytic calculation}

The system is in fact solvable in closed form. Its singular sector splits into the two disjoint chains visible in Fig.~\ref{fig:moving-rule}(b), each generated from one of the two seeds of negative sign, $\beta_1=-\phi_0$ and $\beta_1'=\phi_0-\theta$, by repeated application of the moving rule until the death region is reached. The two projected functions alternate along a chain: with $p_m=((-1)^m+3)/2$, which is $1$ for odd $m$ and $2$ for even $m$, the $m$-th node of the first chain belongs to $f_{p_m}$ and that of the second to $f_{p_{m+1}}$. Since each row of Eq.~\eqref{eq:M2linear} couples a node only to its unique predecessor, a single forward recursion determines all coefficients: the first chain, with angles $\beta_m=-\phi_0-(m-1)\theta$, and the second, with angles $\beta_m'=\phi_0-m\theta$, carry
\begin{align}
G(\beta_m)&=\frac{1}{2\pi\ii}\frac{2}{\ii c_2}\prod_{n=1}^{m}\frac{-\ii c_{p_{n-1}}/2}{|\sin\beta_n|+\ii c_{p_n}/2},\notag\\[2pt]
G(\beta_m')&=\frac{1}{2\pi\ii}\frac{2}{\ii c_1}\prod_{n=1}^{m}\frac{-\ii c_{p_n}/2}{|\sin\beta_n'|+\ii c_{p_{n+1}}/2},
\label{eq:M2chains}
\end{align}
for $m=1,\dots,M_1$ and $m=1,\dots,M_2$, where the prefactors $2/(\ii c_2)$ and $2/(\ii c_1)$ cancel the $n=1$ factors and the lengths of the two chains, fixed by the death region, are
\begin{equation}
M_1=\Bigl\lceil\frac{\pi-\phi_0}{\theta}\Bigr\rceil,\quad
M_2=\Bigl\lceil\frac{\pi+\phi_0}{\theta}\Bigr\rceil-1 .
\label{eq:M2lengths}
\end{equation}
Each chain additionally terminates on one node of the other function, namely the second seed of that function, $\theta-\phi_0\in C_2$ for the first chain and $\phi_0\in C_1$ for the second; it is reached from the chain head through the $+\theta$ branch rather than along the chain, and the two coefficients are
\begin{align}
&G(\theta-\phi_0)=\frac{1}{2\pi\ii}\frac{|\sin\phi_0|}{\bigl(|\sin(\theta-\phi_0)|+\frac{\ii c_2}{2}\bigr)\bigl(|\sin\phi_0|+\frac{\ii c_1}{2}\bigr)},\notag\\[2pt]
&G(\phi_0)=\frac{1}{2\pi\ii}\frac{|\sin(\theta-\phi_0)|}{\bigl(|\sin\phi_0|+\frac{\ii c_1}{2}\bigr)\bigl(|\sin(\theta-\phi_0)|+\frac{\ii c_2}{2}\bigr)}.
\label{eq:M2nodes}
\end{align}
The two chains together exhaust the closure, so Eqs.~\eqref{eq:M2chains}, \eqref{eq:M2lengths} and~\eqref{eq:M2nodes} solve Eq.~\eqref{eq:M2linear} completely; direct substitution confirms the solution.

The closed forms of Eqs.~\eqref{eq:M2chains}, \eqref{eq:M2lengths} and~\eqref{eq:M2nodes} also settle the two features anticipated above; the calculation is elementary but somewhat lengthy. We treat the case at hand, for which $\lfloor\pi/\theta\rfloor$ is even, write $\phi_r=\pi-\lfloor\pi/\theta\rfloor\theta$ for the fold point, and consider $\phi_0=\phi_r-\Delta$ with $\Delta\to0^{+}$; the relevant node of the first chain, $\beta_M$, then tends to $\pi$, with $|\sin\beta_M|=\sin\Delta$. The ghost channel, at the outgoing angle $\pi-\Delta$, exists only for $\Delta>0$; both of its poles lie on the same side of the real axis, so that the two amplitudes of Eq.~\eqref{eq:scat-triples} add coherently,
\begin{equation}
A_{\mathrm g}=\frac{\ii c_1}{2}\,G(\beta_M)+\frac{\ii c_2}{2}\,G(\beta_{M-1}),
\label{eq:ghost-amp}
\end{equation}
$G(\beta_{M-1})$ being the preceding node of that chain. Consecutive nodes are related by Eq.~\eqref{eq:M2chains},
\begin{equation}
G(\beta_M)=-\frac{\ii c_2/2}{|\sin\beta_M|+\ii c_1/2}\,G(\beta_{M-1}),
\label{eq:last-link}
\end{equation}
which makes the two terms of Eq.~\eqref{eq:ghost-amp} cancel except for the numerator, $\frac{\ii c_2}{2}G(\beta_{M-1})|\sin\beta_M|/(|\sin\beta_M|+\ii c_1/2)$, so that
\begin{equation}
A_{\mathrm g}=-\,G(\beta_M)\sin\beta_M\;\simeq\;-\,G(\beta_M)\Delta .
\label{eq:ghost-amp-limit}
\end{equation}
The amplitude of the ghost channel is therefore linear in $\Delta$, and through the normalization of Eq.~\eqref{eq:atomic}---recall that the amplitude sums of Eq.~\eqref{eq:Wr} are the unnormalized weights, $W_r=2\pi^{2}\widetilde W_r$---its probability is
\begin{equation}
W_{\mathrm g}=2\pi^{2}\bigl|G(\beta_M)\bigr|^{2}\sin^{2}\beta_M\simeq C(c)\,\Delta^{2},
\label{eq:ghost-weight}
\end{equation}
which is the general soft form of Eq.~\eqref{eq:soft-open}; the coefficient left unspecified there is given explicitly by the last node of the chain,
\begin{equation}
C(c)=2\pi^{2}\bigl|G(\beta_M)\bigr|^{2}.
\label{eq:ghost-coef}
\end{equation} That coefficient carries the whole $c$ dependence: for equal barrier strengths, $c_1=c_2=c$, Eq.~\eqref{eq:M2chains} gives $C(c)\propto c^{2(M-1)}$ as $c\to0$ and $C(c)\propto c^{-2}$ as $c\to\infty$, so that the strong-coupling law is the universal $c^{-2}$ of Fig.~\ref{fig:fisher}, while the weak-coupling exponent $2(M-1)$ counts the virtual collisions that build the new channel---for $\theta=41^\circ$ it is six, as observed.

The same two coefficients carry the phase. For the channel that continues the classical branch, the two members of Eqs.~\eqref{eq:scat-triples} and~\eqref{eq:triple-plus} take over on the two sides of the fold. Writing the pole denominators of Eq.~\eqref{eq:scat-triples} as limits, the amplitude on each side is
\begin{align}
A_{+}&=\lim_{\epsilon\to0^{+}}\frac{-\ii c_1/2}{\ii\epsilon}\;G(\beta_M)\quad(\Delta>0),\notag\\[2pt]
A_{-}&=-\,\lim_{\epsilon\to0^{+}}\frac{\ii c_2/2}{\ii\epsilon}\;G(\beta_{M-1})\quad(\Delta<0),
\label{eq:branch-amps}
\end{align}
the common divergent factor $1/(\ii\epsilon)$ cancelling in their two-sided ratio, where Eq.~\eqref{eq:last-link} eliminates $\beta_{M-1}$,
\begin{equation}
\frac{\lim_{\Delta\to0^{-}}A_{-}}{\lim_{\Delta\to0^{+}}A_{+}}
=\frac{c_2\,G(\beta_{M-1})}{c_1\,G(\beta_M)}
=-1,
\label{eq:amp-ratio}
\end{equation}
The modulus is therefore continuous across the fold, while the phase differs by $\pi$ on the two sides---the relative phase of Fig.~\ref{fig:phase-shift}. Both features anticipated above, the smooth opening with its scaling and the $\pi$ phase, thus follow from the same two ingredients: the last link of the chain, Eq.~\eqref{eq:last-link}, and the vanishing of $\sin\beta_M$ at the fold.

The chain lengths of Eq.~\eqref{eq:M2lengths} also account for the rational angles. Only the ceilings depend on $\phi_0$: the first length decreases by one when $\phi_0$ passes $\phi_r$ and the second increases by one when it passes $\theta-\phi_r$, so that each fold is exactly where a chain gains or loses a node, and with it one channel. At $\theta=\pi/N$ one has $\phi_r=0$, both folds degenerate onto the barrier directions, and $M_1=M_2=N$ for every $\phi_0$ in the wedge: the closure then contains the same number of poles at every incident angle, and the number of channels stays fixed at $2N$---ten at $\theta=\pi/5$---as seen in Fig.~\ref{fig:channel-map}(a),(b).

\subsection{Summary}

To summarize this section: the projected integral equations of Secs.~\ref{sec:model}--\ref{sec:outgoing}, which reduce the scattering problem to a finite linear system, are solved here in closed form for the two-barrier case, and the exact solution confirms, one by one, the features anticipated from the classical analysis. The number of outgoing channels is not fixed: at a generic angle it changes by one whenever a classical fold is crossed, whereas at the rational angles $\theta=\pi/N$ it stays constant. The new channel opens softly---its weight grows quadratically from zero, Eq.~\eqref{eq:ghost-weight}, and its Fisher information stays finite and constant, Eq.~\eqref{eq:fisher-ghost}---so that the initial-condition sensitivity of the classical scattering leaves no singular trace in the quantum probabilities, which remain smooth through the fold. The sensitivity is not lost, however: it is transferred to the phases, where it survives as the relative phase of $\pi$ between the two channels that continue the classical branches, Eqs.~\eqref{eq:branch-amps} and~\eqref{eq:amp-ratio}. The classical fold structure is thus reproduced in the quantum problem with the same information content, carried by the channel topology and by the phases rather than by the distribution. The same framework is applied to the Galperin billiards in Sec.~\ref{sec:outlook} (Fig.~\ref{fig:galperin-schematic}).

\begin{figure}[t]
\centering
\includegraphics[width=\columnwidth]{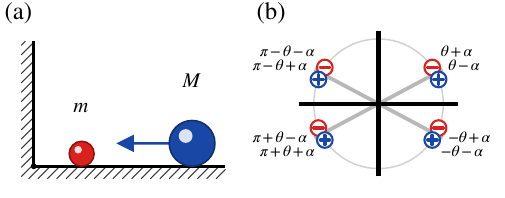}
\caption{(a) The Galperin billiards; (b) its angular representation, with the two $\delta$-barrier lines and the poles of the scattering state.}
\label{fig:galperin-schematic}
\end{figure}

\section{Application to the quantum Galperin billiards}
\label{sec:outlook}

The Galperin billiards \cite{galperin2003} is remarkable in that the number of elastic collisions in a purely mechanical process records the numerical value of $\pi$: a heavy ball of mass $M$ is sent towards a light ball of mass $m$ at rest against a hard wall, and the collisions that occur before the heavy ball reverses its motion add up to $\lfloor\pi/\theta\rfloor$, with $\theta=\arctan\sqrt{m/M}$; for a mass ratio $M/m=b^{2N}$, with $b$ an even integer, this count reproduces the digits of $\pi$ [Fig.~\ref{fig:galperin-schematic}(a)]. The classical problem has since been solved in closed form, including a detailed analysis of the error with which the digits are obtained \cite{aretxabaleta2020math}.

To make the correspondence with the two-particle problem explicit, consider two particles of masses $M$ and $m$ moving on a half-line bounded by a hard wall at the origin, with positions $x_1,x_2\ge0$ and a contact interaction $V=c\,\delta(x_1-x_2)$ between them. Rescaling the coordinates as $X=\sqrt M\,x_1$ and $Y=\sqrt m\,x_2$ makes the kinetic energy isotropic---it becomes that of a free particle of unit mass in the plane---and maps the allowed region onto the first quadrant. The wall then acts as a pair of mirror axes, across which the wave function is antisymmetric. The rescaling acts on the interaction as well: the relative coordinate is a single linear combination of the new variables, and using $\delta(\lambda u)=\delta(u)/|\lambda|$, the $\delta$ function becomes one of the normal distance from the line $x_1=x_2$,
\begin{equation}
\delta(x_1-x_2)=\sqrt{\frac{Mm}{M+m}}\,\delta(\hat n\cdot\mathbf r),\,\, \hat n=\frac{(\sqrt m,-\sqrt M)}{\sqrt{M+m}},
\label{eq:contact-delta}
\end{equation}
so that the contact interaction is a $\delta$-potential barrier along that line, that is, along the direction $\theta=\arctan\sqrt{m/M}$. A plane wave $e^{\ii(k_XX+k_YY)}$ of the rescaled problem corresponds to the two particles carrying momenta $p_1=\sqrt M\,k_X$ and $p_2=\sqrt m\,k_Y$, so that the incident angle is fixed by their ratio,
\begin{equation}
\tan\alpha=\sqrt{\frac{M}{m}}\,\frac{p_2}{p_1}=\sqrt{\frac{m}{M}}\,\frac{v_2}{v_1},
\label{eq:incident-angle}
\end{equation}
a particle at rest corresponding to incidence along the $X$ axis.

\subsection{Quantum Galperin model}

Quantum versions of the Galperin billiards have been discussed in the literature. Brown \cite{brown2020quantum} established an exact isomorphism between the bouncing billiards and Grover's algorithm for quantum search, and a direct quantum-mechanical treatment was given by Cai and Zhang \cite{cai2022}, who argued that the information contained in the classical collision count is transferred to the scattering phase shifts and read the digits of $\pi$ off from those shifts. This model, however, possesses no continuous symmetry, and its barriers are supported on entire rays rather than on compact sets, so that neither a conserved angular momentum nor a separable angular basis is available. To complete the calculation, they resorted to an adiabatic approximation and to an angular quantization of the problem.

These approximations are unnecessary as the model admits an exact solution, which is obtained by the method developed in Secs.~\ref{sec:model}--\ref{sec:outgoing}. To this end, we use the method of images: scattering off the wall is converted into a symmetry of the wave function, and the problem to be solved becomes that of two symmetric intersecting $\delta$-potentials in a space with reflection antisymmetry about the $X$ and $Y$ axes. The initial momentum of the incident wave is then converted, under this symmetry, into four incident components $\phi_j\in\{\alpha,\ \pi+\alpha,\ -\alpha,\ \pi-\alpha\}$, whose coefficients $\kappa_j=+1,+1,-1,-1$ are fixed by the reflection antisymmetry, and each of them into a pair of initial singularities---eight in all, shown in panel (b) of Fig.~\ref{fig:galperin-schematic}.

For this configuration the projected equations of Secs.~\ref{sec:model}--\ref{sec:method} simplify. The barriers lie at $\theta_1=\theta$ and $\theta_2=-\theta$, that is, $\pi-\theta$ modulo $\pi$, with equal strengths $c$, the incident direction is the angle $\alpha$ of Eq.~\eqref{eq:incident-angle}, and the coupled system \eqref{eq:Kform}--\eqref{eq:linear} therefore involves two functions whose poles are generated by the moving rule. The reflection antisymmetry of the state makes the two solutions opposite, $G_2=-G_1$, so that the coupled equations collapse onto a single unknown function; since the moving rule closes that function on a finite set of poles (Appendix~\ref{app:closure}), the equation becomes a finite linear system. The model is thus solved exactly, with no approximation of any kind.

Written for the barrier at $\theta$, the system \eqref{eq:linear} reduces to a single linear system,
\begin{equation}
\Bigl[S_1+\frac{\ii c}{2}\bigl(I-K_{12}\bigr)\Bigr]G_1=F_1^{(0)},
\label{eq:galperin-single-f}
\end{equation}
with $\theta_{12}=2\theta$; here $S_1$, $K_{12}$ and $F_1^{(0)}$ are the matrices and the seed of Eqs.~\eqref{eq:linear}, \eqref{eq:Kmat} and~\eqref{eq:seed}, built on the finite closure of Appendix~\ref{app:closure}.

Throughout this discussion $\alpha$ is taken small, as in the original setting: the classical Galperin problem has the light ball at rest, $v_2=0$, which by Eq.~\eqref{eq:incident-angle} is the case $\alpha=0$. Because the moving rule generates a finite closure (Appendix~\ref{app:closure}), Eq.~\eqref{eq:galperin-single-f} is a finite linear system and can be solved numerically as it stands, without further approximation; its solution gives the poles of the scattering state, and the outgoing map of Sec.~\ref{sec:outgoing} turns them into outgoing atoms.

Solving it gives the outgoing data as in Sec.~\ref{sec:outgoing}, with the normalization $2\pi^2\sum_r\widetilde W_r=\sum_j|\kappa_j|^2=4$ for the four-component incident state. The chain length is fixed by the integer $\lfloor\pi/\theta\rfloor$ that counts the classical collisions, and it is in the phases of the resulting amplitudes that this information reappears, in the closed forms given in the next subsection.

\begin{figure}[t]
\centering
\includegraphics[width=\columnwidth]{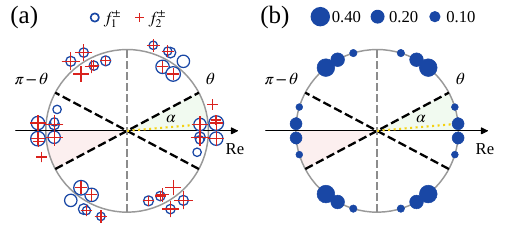}
\caption{(a) Poles generated by the moving rule on the two barriers and (b) the resulting outgoing atoms, for $\theta=28^\circ$, $\alpha=5^\circ$ and $c=1$.}
\label{fig:galperin-solution}
\end{figure}

Figure~\ref{fig:galperin-solution} shows a representative solution, for $\theta=28^\circ$, $\alpha=5^\circ$ and $c=1$. Panel (a) gives the poles generated by the moving rule on the two barriers, and panel (b) the outgoing atoms obtained from them. The distribution carries the symmetry of the state: the atoms form quadruplets $\{\varphi,\pi-\varphi,\pi+\varphi,-\varphi\}$, and no weight leaves in the four directions of the incident wave, $\pm\alpha$ and $\pi\pm\alpha$, which is the fingerprint of the reflection antisymmetry. The weights satisfy $2\pi^2\sum_r\widetilde W_r=4$ to machine precision, as required by the four-component incident state.

\subsection{Scattering phase shifts}

Throughout this subsection $n=\lfloor\pi/\theta\rfloor$ is the number of steps of $2\theta$ that fit in $\pi$, and
\begin{equation}
r=\pi-n\theta
\label{eq:galperin-r}
\end{equation}
is the remainder, $0\le r<\theta$; it is the distance from the last step to the death region.

Equation~\eqref{eq:galperin-single-f} has a structure that can be read off from the matrix itself. Its singular sector is a set of isolated poles, and the moving rule advances a pole by exactly $2\theta$; beginning with a pole that lies in the double region, the move can be repeated, and the angles so generated form a sequence
\begin{equation}
\beta_0\;\longrightarrow\;\beta_0-2\theta\;\longrightarrow\;\cdots\;\longrightarrow\;\beta_0-2\theta\Bigl\lfloor\frac{\beta_0}{2\theta}\Bigr\rfloor,
\label{eq:galperin-chainseq}
\end{equation}
which stops when it falls into the death region. We call such a sequence a chain, and its existence is a property of the equation alone: along a chain every pole is coupled only to its predecessor, so the matrix is bidiagonal and a single forward pass yields the solution.

A chain is labeled by its first pole, for which we write $|[\beta_0]\rangle$; the two chains needed below start from the poles adjacent to the barrier direction, $|[\pi-\theta+\alpha]\rangle$ and $|[\pi-\theta-\alpha]\rangle$, whose representatives lie in the double region [Fig.~\ref{fig:moving-rule}(a)]. Writing a chain as $\beta_0,\beta_1=\beta_0-2\theta,\dots$, the amplitudes are
\begin{equation}
G(\beta_n)=\frac{1}{2\pi\ii}\frac{2}{\ii c}\prod_{k=1}^{n}q_k,\quad q_k=\frac{\ii c/2}{|\sin\beta_k|+\ii c/2},
\label{eq:galperin-chain}
\end{equation}
so that every move contributes a factor of modulus below unity and the weight decays along the chain.

Collecting the weights from the head of a chain, the residue at its far end is
\begin{equation}
Q(\chi)=\prod_{k=0}^{\lfloor\chi/2\theta\rfloor}\frac{\ii c/2}{|\sin(\chi-2k\theta)|+\ii c/2},
\label{eq:galperin-Q}
\end{equation}
for a chain that terminates in the death region, whereas a chain closing on itself gives instead the endpoint combination $S=q_1q_M-q_1-q_M+\prod_{m=1}^{M}q_m$.

The quantity of interest is the ratio of this amplitude to its value as $c\to\infty$, that is, to its value in the classical (hard-wall) limit: since every factor of the product tends to unity in that limit, the ratio is $Q(\chi)$ itself. Its modulus describes the attenuation of the signal by the penetrable barriers, and its phase is the phase shift relative to the hard-wall limit,
\begin{equation}
\Phi_{\chi}=\sum_{k=0}^{\lfloor\chi/2\theta\rfloor}\arctan\frac{2\lvert\sin(\chi-2k\theta)\rvert}{c},
\label{eq:galperin-phase}
\end{equation}
which vanishes as $c\to\infty$, where the classical scattering is recovered, and grows with the number of steps accumulated along the chain. The residues of the poles that fall in $(0,\theta)$ are collected in Table~\ref{tab:galperin-cases}, the upper sign referring to the chain through $\pi-\theta+\alpha$ and the lower to the chain through $\pi-\theta-\alpha$.

\begin{table}[tb]
\caption{Outgoing poles in $(0,\theta)$ and the residues at them; the two signs refer to the two sides $\lambda=\pm1$, and $\chi_\pm := \pi - \theta \pm \alpha$.}
\label{tab:galperin-cases}
\begin{ruledtabular}
\begin{tabular}{lccc}
$\pi/\theta$ & poles & side & residue\tabularnewline
\hline
odd integer & $\theta-\alpha$ & both & $\mp Q(\chi_\pm)/(2\pi\ii)$\tabularnewline
even integer & $\alpha$ & both & $\pm S/(2\pi\ii)$\tabularnewline
generic, $\lfloor\pi/\theta\rfloor$ odd & $\theta-r\mp\alpha$ & $\lambda=+1$ & $\mp Q(\chi_\pm)/(2\pi\ii)$\tabularnewline
generic, $\lfloor\pi/\theta\rfloor$ even & $r\pm\alpha$ & $\lambda=-1$ & $\mp Q(\chi_\pm)/(2\pi\ii)$\tabularnewline
\end{tabular}
\end{ruledtabular}
\end{table}

\paragraph{Non-integer $\pi/\theta$.} When $\pi/\theta$ is not an integer both chains run all the way to the death region, and it is the parity of $n$ that decides where they end. For odd $n$ the two directions are $\theta-r-\alpha$ and $\theta-r+\alpha$, approached from the upper half-plane:
\begin{align}
\Phi_{\theta-r-\alpha}&=\sum_{k=0}^{\lfloor(\pi-\theta+\alpha)/2\theta\rfloor}\arctan\frac{2\lvert\sin((2k+1)\theta-\alpha)\rvert}{c},\\
\Phi_{\theta-r+\alpha}&=\sum_{k=0}^{\lfloor(\pi-\theta-\alpha)/2\theta\rfloor}\arctan\frac{2\lvert\sin((2k+1)\theta+\alpha)\rvert}{c}.
\end{align}
For even $n$ the same two chains end instead at $r+\alpha$ and $r-\alpha$, approached from the lower half-plane:
\begin{align}
\Phi_{r+\alpha}&=\sum_{k=0}^{\lfloor(\pi-\theta+\alpha)/2\theta\rfloor}\arctan\frac{2\lvert\sin((2k+1)\theta-\alpha)\rvert}{c},\\
\Phi_{r-\alpha}&=\sum_{k=0}^{\lfloor(\pi-\theta-\alpha)/2\theta\rfloor}\arctan\frac{2\lvert\sin((2k+1)\theta+\alpha)\rvert}{c}.
\end{align}
In the limit $\alpha\to0$ the two channels carry one and the same phase shift, and the sums above collapse into the single closed form
\begin{equation}
\Phi=\sum_{k=1}^{\lfloor(\pi/\theta+1)/2\rfloor}\arctan\frac{2\sin((2k-1)\theta)}{c}.
\label{eq:galperin-phase-closed}
\end{equation}
For large $c$ every term is confined between its linear and its concave bound, and the elementary identity $\sum_{k=1}^{M}\sin((2k-1)\theta)=\sin^2(M\theta)/\sin\theta$ turns the two sides into the same geometric factor:
\begin{align}
\arctan\Bigl(\frac{2}{c}\Bigr)\frac{\sin^2(M\theta)}{\sin\theta}\;\le\;\Phi\;\le\;\frac{2}{c}\frac{\sin^2(M\theta)}{\sin\theta}
\end{align}
with $M=\Bigl\lfloor\frac{\pi/\theta+1}{2}\Bigr\rfloor$,
so that in this limit the phase shift is fixed by geometry alone, through the factor $\sin^2(M\theta)/\sin\theta$, which grows linearly with the number of steps $M$ for generic angles and saturates whenever $M\theta$ approaches $\pi/2$.

\paragraph{Odd integer $\pi/\theta$.} When $\pi/\theta=N$ is an odd integer both chains end on the single direction $\theta-\alpha$, one on each side of it, and what the scattering measures there is not either chain by itself but the sum of the two sides, $Q(\chi_+)+Q(\chi_-)$ up to a normalization factor; the phase shift is the phase of that sum. In the limit $\alpha\to0$ the two sides coincide and this phase reduces to the closed form \eqref{eq:galperin-phase-closed}, the same expression obtained above for the non-integer case.

\paragraph{Even integer $\pi/\theta$.} When $\pi/\theta=N$ is an even integer, written $N=2M$, the outgoing direction of the two chains coincides exactly with the direction of incidence: both are the angle $\alpha$ itself, so that the scattered channel and the incident channel overlap and interfere. This is visible in the equation as a closed chain, in which no end is left and the two chains drive each other, and the residue is no longer a product of independent steps but the closed-chain combination
\begin{align}
S=q_1q_M-q_1-q_M+\prod_{m=1}^{M}q_m,
\end{align}
with
\begin{equation}
q_k:=\frac{\ii c/2}{\lvert\sin((2k+1)\theta-\alpha)\rvert+\ii c/2},
\end{equation}
and the incident channel adds its own unit amplitude at this same direction, so that the two interfere and the phase shift is the argument of the sum:
\begin{equation}
\Phi_{\alpha}=\arg(1+S).
\end{equation}

The phase shifts carry one further piece of information. In the classical (hard-wall) limit the phase shift of an outgoing direction is fixed by the number of collisions that lead to it: every reflection contributes a factor $-1$, so that the accumulated phase is $\pi$ for an odd number of collisions and $0$ for an even one. A barrier of finite strength leaves this counting untouched---the chains and their lengths are fixed by geometry---and superimposes on it the continuous correction of Eq.~\eqref{eq:galperin-phase}. The phase shift is therefore an explicit function of both the interaction strength $c$ and the opening angle $\theta$, and the relation between the two can be extracted from it exactly; since $\theta=\arctan\sqrt{m/M}$ is fixed by the mass ratio, a measured phase shift determines $c$ for a system of given masses. In this way the scattering phase shifts provide a direct and exact means of measuring the interaction strength \cite{gustafson2021phaseshifts}.

\section{Conclusion and outlook}
\label{sec:conclusion}

We have formulated a systematic method for the quantum scattering of a plane wave by intersecting $\delta$-potential barriers in two dimensions, the boundary-free descendant of the Gaudin kaleidoscope. Motivated by the observation that BA solvability in the bounded kaleidoscope is a composite criterion---the KYBE is necessary but not sufficient, and the Bethe-ansatz equations require boundary conditions---we studied whether the specialness of the angles $\pi/N$ survives when the boundaries are removed. The moving-rule method reduces the singular part of the projected Lippmann--Schwinger equations to a finite-dimensional linear system, and the regular part is obtained by a regular integral equation; the outgoing state is an atomic probability measure whose normalized weights obey $\sum_rW_r=1$, while the unnormalized weights satisfy the closed form $2\pi^2\sum_r\widetilde W_r=\sum_j\lvert\kappa_j\rvert^2$, the total weight of the incident components---unity for a single plane wave and four for the four-component incident state of Sec.~\ref{sec:outlook}. The sensitivity of the distribution to the incident angle follows the classical channel structure: the number of outgoing channels is fixed at the rational angles $\theta=\pi/N$, while at generic angles channels are created and destroyed as the critical directions are crossed. This sensitivity is invisible in the scattering probabilities, where the distribution stays smooth, and is carried instead by the scattering phase shifts.

The same framework solves the quantum Galperin billiards exactly. The method of images turns the two hard walls into a reflection antisymmetry of the wave function, which collapses the two coupled projected equations into a single linear system and thereby organizes the poles of the scattering state into chains. Each chain is solved by one forward pass, and the outgoing distribution together with the scattering phase shifts relative to the hard-wall limit follows in closed form, with no adiabatic, semiclassical or partial-wave input. The phase shifts reproduce the classical channels: at generic angles the two chains terminate in the death region, at odd integer ratios they meet on a single direction where the two sides must be added, and at even integer ratios the chain closes on itself so that the scattered and incident channels interfere in the same direction. Since $\theta=\arctan\sqrt{m/M}$, the phase shifts are explicit functions of the mass ratio. The connection of these phases to Fisher information and precision measurement \cite{ding2022criticalmetrology} will be presented elsewhere. The same geometry also invites a description in terms of the complexity and scrambling diagnostics developed for quantum billiards and many-body systems \cite{camargo2024krylov,xu2024scrambling}.

Two questions remain open. The moving-rule method has been developed here for barriers that meet at a single point; whether it extends to barriers intersecting at several points, or not meeting at all, is unclear, since the finite pole closure on which the method rests has been established only for the common-vertex geometry. Barriers supported on a more general union of lines belong to the class of ``leaky'' quantum structures \cite{exner2003approx,exner2008review}, and the coincidence limit of several $\delta$ potentials has recently been re-examined \cite{loran2022multidelta}. The analysis has also been restricted throughout to repulsive barriers, $c>0$, for which the scattering state is the only relevant solution at positive energy. For $c<0$ the barriers become attractive and support bound states, and the scattering problem acquires, alongside them, the resonances associated with these states. Attractive zero-range potentials tie the phase shift to the bound-state spectrum through Levinson's theorem \cite{camblong2019levinson,belloni2014delta}, and in few-body settings they produce the impurity bound states and Efimov structures studied in Refs.~\cite{shi2016boundstates,naidon2017efimov}. How such resonances interact with the collision structure identified here is an open question.

\appendix
\begin{acknowledgments}
Y.-C.~Yu. thanks Y.-B.~Zhang for helpful discussions. This work was supported by the National Natural Science Foundation of China under Grant Nos.~12274419, the Quantum Science and Technology-National Science and Technology Major Project under Grant No.~2023ZD0300404, the CAS Project for Young Scientists in Basic Research under Grant No.~YSBR-055, and the National Key R\&D Program of China under Grant No.~2022YFA1404104.
\end{acknowledgments}

\section*{Data availability}
The computational codes used to generate the results of this work are
publicly available at
\url{https://github.com/YiCongYu-CAS/intersecting-delta-barrier-scattering}.

\section{Derivation of the moving rule}
\label{app:movingrule}

Here we derive the moving rule of Eq.~\eqref{eq:rules}, following the pole
representation of Eq.~\eqref{eq:pole}. We first write the kernel in
explicit form, then evaluate its action on a single pole by contour
integration, and finally read off the displacement and the residue of the
generated poles. Throughout we set $a=1$ and use the
complexified projected momentum $\zeta$ of Eq.~\eqref{eq:pole}; real-axis
integrals are understood as $\zeta$ taken on the real axis, and
$s:=\operatorname{sign}(\sin\beta)$ denotes the side of the pole.

At unit energy, the kernel of Eq.~\eqref{eq:kernel} acts on a function $f=f(\zeta)$ as
\begin{align}
\label{eq:app:kernel}
&(\hat I_y\hat G_0\hat R_\theta f)(\zeta)
= \frac{1}{2\sqrt{1-\zeta^2}}\int\dd\zeta'\,f(\zeta')\notag\\
&\quad \times\Bigl[\frac{1}{\zeta'-\Gamma^-_\theta(\zeta)+\iz}
 - \frac{1}{\zeta'-\Gamma^+_\theta(\zeta)-\iz}\Bigr],
\end{align}
with $\Gamma^\pm_\theta(\zeta)=\zeta\cos\theta\pm\sqrt{1-\zeta^2}\,\lvert\sin\theta\rvert$
as in Eq.~\eqref{eq:Gamma}. The small imaginary parts place the pole of the
first term at $\zeta'=\Gamma^-_\theta(\zeta)-\iz$, i.e.\ in the lower
half-plane, and that of the second term at
$\zeta'=\Gamma^+_\theta(\zeta)+\iz$ in the upper half-plane. For the kernel
angles $\theta_{km}=\lvert\theta_k-\theta_m\rvert$ of the problem,
$\theta\in[0,\pi)$ so that $\lvert\sin\theta\rvert=\sin\theta$, and with
$\alpha:=\arccos\zeta\in[0,\pi]$ one has
$\Gamma^\pm_\theta(\zeta)=\cos(\alpha\mp\theta)$, a form used below.

We evaluate this kernel on a single pole by taking $f=\ket{\beta}$, whose boundary value on the real axis is
\begin{align}
\ket{\beta}(\zeta')
= \frac{1}{\zeta'-\cos\beta-\iz\,s},
\quad
s:=\operatorname{sign}(\sin\beta),
\end{align}
with the pole at $\zeta'=\cos\beta+\iz\,s$, in the upper half-plane for
$s=+1$ and in the lower one for $s=-1$. We close the integration contour in
Eq.~\eqref{eq:app:kernel} in the upper half-plane, with $\zeta$ fixed as the
point where the output is evaluated, and examine the two terms of the
bracket separately. The first term contributes
\begin{align}
J_- = \int\dd\zeta'\,
\frac{1}{\zeta'-\cos\beta-\iz s}\,
\frac{1}{\zeta'-\Gamma^-_\theta(\zeta)+\iz}.
\end{align}
For $s=+1$ the pole of $\ket{\beta}$ lies in the upper half-plane while that
of the kernel factor lies in the lower one, so the contour encloses only the
former and $J_-=2\pi\ii\,(\cos\beta-\Gamma^-_\theta(\zeta))^{-1}$; for $s=-1$
both poles lie below the axis and $J_-=0$. The second term contributes
\begin{align}
J_+ = \int\dd\zeta'\,
\frac{1}{\zeta'-\cos\beta-\iz s}\,
\frac{1}{\zeta'-\Gamma^+_\theta(\zeta)-\iz},
\end{align}
whose kernel pole lies in the upper half-plane. For $s=+1$ the two enclosed
residues, $(\cos\beta-\Gamma^+_\theta(\zeta))^{-1}$ and
$(\Gamma^+_\theta(\zeta)-\cos\beta)^{-1}$, cancel, hence $J_+=0$; for $s=-1$
only the kernel pole is enclosed and
$J_+=2\pi\ii\,(\Gamma^+_\theta(\zeta)-\cos\beta)^{-1}$. Combining the two
terms, the full integral is nonvanishing only through the pole that shares
the half-plane of $\ket{\beta}$, and reads
\begin{align}
\label{eq:app:contour}
&\int\dd\zeta'\,\ket{\beta}(\zeta')
\Bigl[
\frac{1}{\zeta'-\Gamma^-_\theta(\zeta)+\iz}
-\frac{1}{\zeta'-\Gamma^+_\theta(\zeta)-\iz}
\Bigr]\notag\\
&\quad = 2\pi\ii\,
\bigl(\Gamma^{-\operatorname{sign}(\sin\beta)}_\theta(\zeta)-\cos\beta\bigr)^{-1},
\end{align}
where $-\operatorname{sign}(\sin\beta)$ selects $\Gamma^-$ for $s=+1$ and
$\Gamma^+$ for $s=-1$, and the identity holds for $\zeta$ on the energy
shell, $|\zeta|\le1$. Inserting Eq.~\eqref{eq:app:contour} into
Eq.~\eqref{eq:app:kernel} and using the definition of
$\hat K_\theta$ in Eq.~\eqref{eq:Kform}, we obtain the exact identity
\begin{align}
\label{eq:app:Kbeta}
\hat K_\theta\ket{\beta}(\zeta)
= \frac{1}{\sqrt{1-\zeta^2}}\,
\frac{1}{\Gamma^{-\operatorname{sign}(\sin\beta)}_\theta(\zeta)-\cos\beta}.
\end{align}
The vanishing contributions in the two cases ($J_+=0$ for $s=+1$ and $J_-=0$
for $s=-1$) arise from the cancellation of two enclosed residues, not from
the pole lying outside the contour.

Equation~\eqref{eq:app:contour} shows that a pole with $s=+1$
($\sin\beta>0$) moves only through the $\Gamma^-$ channel and one with
$s=-1$ only through the $\Gamma^+$ channel. With
$\Gamma^\pm_\theta(\zeta)=\cos(\alpha\mp\theta)$ and
$\alpha=\arccos\zeta\in[0,\pi]$, the singularity condition
$\Gamma^{\mp}_\theta(\zeta)=\cos\beta$ reads
$\cos(\alpha\pm\theta)=\cos\beta$ for the two channels, i.e.\
$\alpha\pm\theta=\pm\beta+2\pi n$. Under the constraint $\alpha\in[0,\pi]$
each channel admits two solutions,
\begin{align}
s=+1:&\quad \alpha=\beta-\theta \ \ \text{or}\ \ 2\pi-\beta-\theta,\\
s=-1:&\quad \alpha=2\pi-\beta+\theta \ \ \text{or}\ \ \beta+\theta-2\pi,
\end{align}
which give only two locations $\zeta_\ast=\cos\alpha$ on the energy shell,
\begin{align}
\label{eq:app:moves}
\zeta_\ast&=\cos(\beta-\theta)\quad(\text{channel A}),\notag\\
\zeta_\ast&=\cos(\beta+\theta)\quad(\text{channel B}).
\end{align}
The corresponding feasibility conditions combine into
$\sin\beta\,\sin(\beta-\theta)>0$ for channel A and
$\sin\beta\,\sin(\beta+\theta)<0$ for channel B. Geometrically, $s$ marks
the half-circle of the pole and the real axis is the branch cut of the
$\ket{\beta}$ notation: channel A rotates the pole clockwise by $\theta$
without crossing the real axis, whereas channel B rotates it
anticlockwise across the axis.

The residue of the displaced pole is obtained by differentiating the channel function $\Gamma^{\mp}_\theta(\zeta)=\cos(\arccos\zeta\mp\theta)$ at $\zeta=\zeta_\ast$:
\begin{align}
\frac{\dd\Gamma^{\mp}_\theta}{\dd\zeta}\bigg|_{\zeta=\zeta_\ast}
= \frac{\sin(\arccos\zeta_\ast\mp\theta)}{\sqrt{1-\zeta_\ast^2}}
= \frac{\pm\sin\beta}{\sqrt{1-\zeta_\ast^2}},
\end{align}
where $\arccos\zeta_\ast\mp\theta=\pm\beta$ (mod $2\pi$). Hence the singular
part of Eq.~\eqref{eq:app:Kbeta} near $\zeta_\ast$ is
\begin{align}
\frac{1}{\sqrt{1-\zeta^2}}\,
\frac{1}{\Gamma^{\mp}_\theta(\zeta)-\cos\beta}
\;\approx\;
\frac{1}{\pm\sin\beta}\,\frac{1}{\zeta-\zeta_\ast},
\end{align}
which assigns $+1/\lvert\sin\beta\rvert$ to channel A and
$-1/\lvert\sin\beta\rvert$ to channel B. The common factor
$1/\lvert\sin\beta\rvert$ is the Jacobian of the angle parametrization
$\beta\mapsto\cos\beta$, and the sign distinguishes the move that crosses
the branch cut of the $\ket{\beta}$ notation from the one that does not.
Combining the displacements of Eq.~\eqref{eq:app:moves} with these residues,
and writing the result as the extraction of the singular part, gives exactly
the moving rule of Eq.~\eqref{eq:rules}. In particular, for
$0<\beta<\theta$ or $\pi<\beta<\pi+\theta$ both Heaviside conditions fail,
$\hat K_\theta\ket{\beta}$ is purely regular, and no new pole is produced:
this is the death region of Fig.~\ref{fig:moving-rule}, which guarantees
that the closure generated by repeated moves is finite.

\section{Finiteness of the pole closure for arbitrary configurations}
\label{app:closure}

In Sec.~\ref{sec:method} the closure of the seeds under the moving rules is defined by iterating the simultaneous replacement of Eq.~\eqref{eq:closure} until no further poles are produced. Every pole originates from the incident wave: the scattering off barrier $k$ creates the seed pair $\{\pm\theta_{0k}\}$, and all these $\delta$ contributions form the initial set $C^{(0)}:=\bigcup_{k=1}^{M}\{\pm\theta_{0k}\}$. Write $T_\theta$ for the map induced by the moving rule Eq.~\eqref{eq:rules}: each allowed move advances an angle by $\pm\theta$, and in the death regions no new angle is produced. Let $\Theta:=\{\theta_{km}:1\le k<m\le M\}$ be the set of all differences between barrier directions. Since the Heaviside conditions and the generated angles depend only on the angle and on the difference $\theta_{km}$, the projected function to which a pole belongs does not affect which new angles are produced; the closure can therefore be followed as a single nested sequence of finite sets $C^{(0)}\subseteq C^{(1)}\subseteq\cdots$, where $C^{(n+1)}$ is obtained from $C^{(n)}$ by the simultaneous replacement of Eq.~\eqref{eq:closure}: every angle added at step $n+1$ is of the form $T_\theta(\beta)$ for some $\beta\in C^{(n)}$ and some $\theta\in\Theta$. Our aim in this appendix is to prove that, for every barrier configuration --- an arbitrary number $M$ of barriers and arbitrary orientation angles $\theta_k$ --- this sequence stabilizes after finitely many steps: there exists a finite $N$ with $C^{(N+1)}=C^{(N)}$, so that the limit $C^{(\infty)}=\bigcup_{n\geq0}C^{(n)}$ is the closure computed in Sec.~\ref{sec:method}. Because $C^{(0)}$ is finite and each step adds only finitely many new angles, the stabilized limit is then automatically finite; this is the finiteness asserted in Sec.~\ref{sec:method}.

To prove this stabilization, we first prove a lemma: an increasing sequence of sets generated from a finite set by finitely many maps either stops growing after finitely many steps, or else contains an infinite orbit in which every step moves to a new point.

\begin{lemma}[Infinite nontrivial orbits in finite systems of maps]
\label{lem:infinite-orbit}
Let $X$ be a set and let $f_0=\operatorname{id}_X$ together with $f_1,\ldots,f_M\colon X\to X$ be finitely many maps. For $C\subseteq X$ define the set operator $\Phi(C)=\bigcup_{m=0}^{M}f_m(C)$. Let $A\subseteq X$ be a nonempty finite set, and define recursively $A_0=A$ and $A_{n+1}=\Phi(A_n)$ for $n\geq0$. Since $f_0=\operatorname{id}_X$, one has $A_n\subseteq A_{n+1}$ for every $n$. If, moreover, the sequence grows strictly at every step, $A_n\subsetneq A_{n+1}$ for every $n\geq0$, then there exists an infinite sequence $a_0,a_1,a_2,\ldots$ in $X$ together with indices $i_0,i_1,i_2,\ldots\in\{1,\ldots,M\}$ such that, for every $n\geq0$,
\begin{align}
a_{n+1}=f_{i_n}(a_n)
\quad\text{and}\quad
a_{n+1}\neq a_n .
\end{align}
\end{lemma}

\begin{proof}
We first characterize the sets $A_n$: $A_n$ consists precisely of the points that can be reached from some element of $A$ by at most $n$ applications of the maps $f_1,\ldots,f_M$. More precisely,
\begin{align*}
A_n
&=\Bigl\{f_{j_k}\circ\cdots\circ f_{j_1}(a):
a\in A,\ 0\leq k\leq n,\\
&\quad\quad
j_1,\ldots,j_k\in\{1,\ldots,M\}\Bigr\},
\end{align*}
where for $k=0$ the composition is understood as the identity, so that the displayed element is simply $a$. The identity is proved by induction on $n$. For $n=0$ both sides equal $A$. If the claim holds for $n$, then, using $A_{n+1}=\bigcup_{m=0}^{M}f_m(A_n)$ together with $f_0=\operatorname{id}_X$, the set $A_{n+1}$ is the union of $A_n$ (the term $m=0$) with all points obtained by appending one more application of some $f_m$, $1\leq m\leq M$, to a point already in $A_n$; this is exactly the characterization for $n+1$.

We next show that nontrivial finite orbits of arbitrary length exist. Fix an integer $n\geq1$. By the strict-growth assumption we may choose an element $x_n\in A_n\setminus A_{n-1}$. Since $x_n\in A_n$, the characterization just established provides $a_0\in A$ and indices $j_0,\ldots,j_{n-1}\in\{0,1,\ldots,M\}$ such that the recursion
\begin{align*}
a_{k+1}=f_{j_k}(a_k),\quad 0\leq k<n,
\end{align*}
terminates at $a_n=x_n$. We claim that every step of this orbit moves to a different point, $a_{k+1}\neq a_k$ for all $0\leq k<n$. Suppose, to the contrary, that $a_{k+1}=a_k$ for some $k$. Then the $k$-th application did not change the current point, so it may be deleted from the sequence of maps without altering the final result: starting from $a_0\in A$, the point $x_n=a_n$ is reached in at most $n-1$ steps. By the same characterization this implies $x_n\in A_{n-1}$, contradicting $x_n\in A_n\setminus A_{n-1}$. Hence $a_{k+1}\neq a_k$ for every $0\leq k<n$. Moreover, since $f_0(a_k)=a_k$ leaves every point fixed, an index that produces a genuine move cannot be $0$; therefore $j_k\in\{1,\ldots,M\}$ for all $k$. By the arbitrariness of $n$, the system contains nontrivial finite orbits of every length $n\geq1$: sequences $a_0\to a_1\to\cdots\to a_n$ with $a_0\in A$, $a_{k+1}=f_{j_k}(a_k)$, $j_k\in\{1,\ldots,M\}$, and $a_{k+1}\neq a_k$ for $0\leq k<n$.

We now organize these orbits into a rooted tree $T$. The root is denoted by $r$. Every other node of $T$ is a finite sequence $(a_0,a_1,\ldots,a_n)$, $n\geq0$, satisfying: (i) $a_0\in A$; (ii) for each $0\leq k<n$ there is an index $i_k\in\{1,\ldots,M\}$ with $a_{k+1}=f_{i_k}(a_k)$; and (iii) $a_{k+1}\neq a_k$ for each $0\leq k<n$. The parent--child relation is the extension relation between finite sequences: $(a_0,\ldots,a_n)$ is the parent of $(a_0,\ldots,a_n,a_{n+1})$ whenever the longer sequence again satisfies the conditions above. The children of the root $r$ are precisely the one-element sequences $(a_0)$ with $a_0\in A$.

The tree $T$ is finitely branching: since $A$ is finite, the root has only finitely many children, and since there are only finitely many maps $f_1,\ldots,f_M$, every non-root node $(a_0,\ldots,a_n)$, whose children are among $(a_0,\ldots,a_n,f_i(a_n))$ with $f_i(a_n)\neq a_n$, has at most $M$ children.

Moreover, $T$ is infinite: by the preceding construction, for every $n\geq1$ there exists a nontrivial finite orbit of length $n$, hence a node of $T$ at depth $n$; thus $T$ contains nodes at arbitrarily large depths and therefore infinitely many nodes.

Finally, K\"onig's infinity lemma asserts that every infinite, finitely branching rooted tree contains an infinite path \cite[Lemma 8.1.2, p.~215]{Diestel2017GraphTheory}; by the two preceding paragraphs $T$ is both, so it contains an infinite path
\begin{align*}
r,\quad (a_0),\quad (a_0,a_1),\quad (a_0,a_1,a_2),\quad \ldots
\end{align*}
These finite sequences are mutually compatible and therefore determine an infinite sequence $a_0,a_1,a_2,\ldots$. By the definition of $T$, for each $n\geq0$ there exists an index $i_n\in\{1,\ldots,M\}$ such that $a_{n+1}=f_{i_n}(a_n)$ and $a_{n+1}\neq a_n$. This completes the proof.
\end{proof}

\begin{figure}[t]
\centering
\includegraphics[width=\columnwidth]{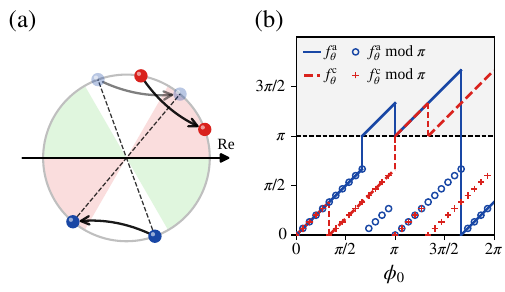}
\caption{
Schematic examples of the two elementary moves of the moving rule, for
$\theta=\pi/3$.
(a) Circle of pole angles: a clockwise move $\beta\mapsto\beta-\theta$
starting in the upper half-plane (red), one starting in the lower
half-plane (blue), and the mirror of the latter under the central symmetry
$\beta\mapsto\beta+\pi$ (translucent blue); dashed diameters join each pole
of the lower example to its mirror, and light red/green shading marks the
death and double regions of Fig.~\ref{fig:moving-rule}. A move in the lower
half-plane is thereby mapped onto a move in the upper half-plane, so that
it suffices to follow the upper half-plane alone.
(b) The maps $f_\theta^{\mathrm a}$ and $f_\theta^{\mathrm c}$ defined
below, shown as functions of the angle $\phi_0$ (solid and dashed curves),
together with their reductions modulo $\pi$ (open circles and crosses),
which fold into $[0,\pi)$; the light shading marks the region
$\phi_0>\pi$. After this folding both maps are nonincreasing in the reduced
angle $\rho=\phi_0\bmod\pi$, every nontrivial move lowering $\rho$ by at
least $\min(\theta,\pi-\theta)$ [Eq.~\eqref{eq:app:mono}]---a monotonicity
that holds for every angle $\theta\in(0,\pi)$ and underpins the finiteness
of the closure.
}
\label{fig:moving-examples}
\end{figure}

In order to study the moving rule Eq.~\eqref{eq:rules} as an iteration of maps on the circle of angles, we isolate its two elementary displacements (Fig.~\ref{fig:moving-examples}): the anticlockwise move $f_\theta^{\mathrm a}$ and the clockwise move $f_\theta^{\mathrm c}$. For fixed $\theta\in[0,\pi)$ and $\phi_0\in[0,2\pi)$ we define
\begin{align}
\label{eq:app:fa}
f_\theta^{\mathrm a}(\phi_0)
&=
\begin{cases}
\phi_0, & \phi_0\in[0,\pi-\theta),\\
\phi_0+\theta, & \phi_0\in[\pi-\theta,\pi),\\
\phi_0, & \phi_0\in[\pi,2\pi-\theta),\\
\phi_0+\theta-2\pi, & \phi_0\in[2\pi-\theta,2\pi),
\end{cases}\\[3mm]
\label{eq:app:fc}
f_\theta^{\mathrm c}(\phi_0)
&=
\begin{cases}
\phi_0, & \phi_0\in[0,\theta),\\
\phi_0-\theta, & \phi_0\in[\theta,\pi),\\
\phi_0, & \phi_0\in[\pi,\pi+\theta),\\
\phi_0-\theta, & \phi_0\in[\pi+\theta,2\pi).
\end{cases}
\end{align}
The shifted sectors are exactly the activation sets of the two Heaviside conditions in Eq.~\eqref{eq:rules}: $f_\theta^{\mathrm a}$ adds $\theta$ where $\sin\phi_0\sin(\phi_0+\theta)<0$, and $f_\theta^{\mathrm c}$ subtracts $\theta$ where $\sin\phi_0\sin(\phi_0-\theta)>0$; on the remaining sectors each map is the identity. Angles are understood modulo $2\pi$, so in the last sector of $f_\theta^{\mathrm a}$ the value $\phi_0+\theta$, which would leave $[0,2\pi)$, is represented by $\phi_0+\theta-2\pi\in[0,\theta)$. Since $0\le\theta<\pi$, the four sectors of each map are disjoint and cover $[0,2\pi)$, the maps are well defined and reduce to the identity for $\theta=0$, and the sector boundaries, where the sine product vanishes, are immaterial, being understood by one-sided limits as in Appendix~\ref{app:Q}.

We record a monotonicity property of these maps, which will be the engine of the stabilization argument below. For $x\in[0,2\pi)$ write $\rho(x):=x\bmod\pi\in[0,\pi)$ for the angle reduced modulo $\pi$, and set $d_\theta:=\min(\theta,\pi-\theta)$. Both maps are nonincreasing in the reduced coordinate: for $h\in\{f_\theta^{\mathrm a},f_\theta^{\mathrm c}\}$ one has $\rho(h(x))\le\rho(x)$ for every $x$, and whenever the move is nontrivial, $h(x)\neq x$, the reduced angle drops by at least $d_\theta$,
\begin{align}
\label{eq:app:mono}
\rho\bigl(h(x)\bigr)\le\rho(x)-d_\theta\quad(h(x)\neq x).
\end{align}
Indeed, on its active sectors $f_\theta^{\mathrm a}$ crosses the $\pi$ boundary and lowers $\rho$ by $\pi-\theta$, while $f_\theta^{\mathrm c}$ lowers it by $\theta$; both amounts are at least $d_\theta$, and on the complementary sectors both maps act as the identity, leaving $\rho$ unchanged.

Starting from the same initial set $C^{(0)}$, we now generate a larger single sequence by acting, at every step, on \emph{every} angle with the moving rules of \emph{all} differences simultaneously. With
\begin{align}
\label{eq:app:big-iter}
\widetilde C^{(0)}&:=C^{(0)},\notag\\
\Phi(S)&:=S\cup\bigcup_{\theta\in\Theta}T_\theta(S),\notag\\
\widetilde C^{(n+1)}&:=\Phi\bigl(\widetilde C^{(n)}\bigr),
\end{align}
every angle of $\widetilde C^{(n)}$ is advanced by every $T_\theta$, $\theta\in\Theta$, at each step, whereas in the closure sequence an angle is moved only by the differences with which it enters the coupled equations. Since the closure step only ever applies differences $\theta\in\Theta$ to angles of $C^{(n)}$, induction on $n$ shows $C^{(n)}\subseteq\widetilde C^{(n)}$ for every $n$: every angle added to $C^{(n+1)}$ is of the form $T_\theta(\beta)$ with $\beta\in C^{(n)}\subseteq\widetilde C^{(n)}$ and $\theta\in\Theta$, hence lies in $T_\theta(\widetilde C^{(n)})\subseteq\widetilde C^{(n+1)}$.

We prove the stronger statement that the enlarged sequence $\widetilde C^{(n)}$ stabilizes after finitely many steps, so that every $\widetilde C^{(n)}$ is finite; together with the inclusion $C^{(n)}\subseteq\widetilde C^{(n)}$ this yields the finiteness of the closure $C^{(\infty)}$ directly, which is the desired result of this appendix.

We argue by contradiction. A vanishing difference $\theta_{km}=0$, which occurs for two parallel barriers, makes $T_0$ and both maps $f_0^{\mathrm a}$, $f_0^{\mathrm c}$ the identity: they generate no new angle, and being already represented by the identity in the family they may be kept without changing $\Phi$. Genuine moves are therefore produced only by nonzero differences; write $\Theta_*:=\{\theta\in\Theta:\theta>0\}$. If $\Theta_*$ is empty there are no genuine moves at all and the enlarged sequence is constant after the first step, so the claim is trivial. Otherwise set
\begin{align*}
\delta:=\min_{\theta\in\Theta_*}d_\theta=\min_{\theta\in\Theta_*}\min(\theta,\pi-\theta)>0 .
\end{align*}
Suppose that $\widetilde C^{(n)}$ never stabilizes. Being increasing, the sequence then grows strictly at every step, $\widetilde C^{(n)}\subsetneq\widetilde C^{(n+1)}$ for all $n$. The maps $f_\theta^{\mathrm a}$ and $f_\theta^{\mathrm c}$, $\theta\in\Theta$, together with the identity, form a finite family of maps of the circle of angles, and the recursion $\widetilde C^{(n+1)}=\Phi(\widetilde C^{(n)})$ is the iteration of exactly this family in the sense of Lemma~\ref{lem:infinite-orbit} (the identity accounts for the term $S$ in $\Phi$, and $f_\theta^{\mathrm a}(S)\cup f_\theta^{\mathrm c}(S)=S\cup T_\theta(S)$). Lemma~\ref{lem:infinite-orbit} therefore supplies an infinite orbit $a_0,a_1,a_2,\ldots$ in which every step is a genuine move, $a_{n+1}\neq a_n$, performed by one of the maps $f_\theta^{\mathrm a}$ or $f_\theta^{\mathrm c}$; the zero-difference maps are the identity and cannot appear in such a step, so the acting maps carry indices $\theta\in\Theta_*$. By the monotonicity property of Eq.~\eqref{eq:app:mono} each such step lowers the reduced angle by at least $\delta$,
\begin{align*}
\rho(a_{n+1})\le\rho(a_n)-\delta,
\end{align*}
so that $\rho(a_n)\le\rho(a_0)-n\delta$ for every $n$. The right-hand side becomes negative for sufficiently large $n$, contradicting $\rho(a_n)\ge0$. Hence $\widetilde C^{(n)}$ stabilizes after finitely many steps: there is a finite $N$ with $\widetilde C^{(N+1)}=\widetilde C^{(N)}$, so that $\widetilde C^{(\infty)}=\widetilde C^{(N)}$ is finite. Since $C^{(n)}\subseteq\widetilde C^{(n)}$ for every $n$, the closure $C^{(\infty)}=\bigcup_{n\geq0}C^{(n)}\subseteq\widetilde C^{(\infty)}$ is finite as well. This completes the proof of the finiteness of the pole closure asserted in Sec.~\ref{sec:method}.

\section{Closed forms of the regularized remainder}
\label{app:Q}

For completeness we record here the explicit piecewise forms of the
regularized remainder
$Q_{\theta,\beta}(z):=s(z)\bigl(\hat K_\theta\ket{\beta}-\mathcal S_\theta\ket{\beta}\bigr)$
introduced in Eq.~\eqref{eq:Qdef}, obtained by the analytic cancellation
described in Sec.~\ref{sec:method}. The formulas hold for every barrier angle
$\theta\in(0,\pi)$ and every pole angle $\beta$ with $\sin\beta\neq0$.
Throughout, $s(z):=\sqrt{1-z^2}$, $\alpha:=\arccos z\in[0,\pi]$,
$s_\beta:=\operatorname{sign}(\sin\beta)$, and we use the abbreviations
$\theta_{<}:=\min(\theta,\pi-\theta)$ and $\theta_{>}:=\max(\theta,\pi-\theta)$. Sector
boundaries are understood by one-sided limits, so the value at
$\beta=\theta_{<},\theta_{>}$ or $\sin\beta=0$ is immaterial.

For a pole in the upper half-plane, $0<\beta<\pi$,
\begin{widetext}
\begin{align}
\label{eq:appQ:upper}
Q_{\theta,\beta}(z)
&=
\begin{cases}
\displaystyle
\frac{1}{\cos(\alpha+\theta)-\cos\beta},
& 0<\beta<\theta_{<} \ \ \text{(death region)},\\[3mm]
\displaystyle
\frac{\sin(\alpha+\beta)}
{2\sin\beta\;\sin\frac{\alpha+\beta-\theta}{2}\;\sin\frac{\alpha+\beta+\theta}{2}},
& \theta_{<}<\beta<\theta_{>},\ \ \theta\le\frac{\pi}{2}
\ \ \text{(single move } \beta-\theta\text{)},\\[5mm]
\displaystyle
\frac{\sin(\beta-\alpha)}
{2\sin\beta\;\sin\frac{\alpha-\beta-\theta}{2}\;\sin\frac{\alpha+\theta-\beta}{2}},
& \theta_{<}<\beta<\theta_{>},\ \ \theta>\frac{\pi}{2}
\ \ \text{(single move } \beta+\theta\text{)},\\[5mm]
\displaystyle
\frac{1}{\cos(\alpha-\theta)-\cos\beta},
& \theta_{>}<\beta<\pi \ \ \text{(crossing region)}.
\end{cases}
\end{align}
\end{widetext}
The four sectors correspond to the three output patterns of the moving
rule~Eq.~\eqref{eq:rules}: in the death region no pole is produced and
the remainder equals the full action; in the single-move regions the pole
generated by the active move, $\beta-\theta$ for $\theta\le\pi/2$ or
$\beta+\theta$ for $\theta>\pi/2$, is subtracted with its exact residue; in
the crossing region the subtracted pole has moved across the real axis. For
$\theta\le\pi/2$ one has $\theta_{<}=\theta$, $\theta_{>}=\pi-\theta$, and the middle two rows
collapse to the single sector $\theta<\beta<\pi-\theta$.

For a pole in the lower half-plane, $\pi<\beta<2\pi$, it is convenient to
set $\gamma:=2\pi-\beta\in(0,\pi)$, so that $\cos\beta=\cos\gamma$ and
$\sin\beta=-\sin\gamma$; then
\begin{widetext}
\begin{align}
\label{eq:appQ:lower}
Q_{\theta,\beta}(z)
&=
\begin{cases}
\displaystyle
\frac{1}{\cos(\alpha-\theta)-\cos\gamma},
& \theta_{>}<\gamma<\pi \ \ \text{(death region)},\\[3mm]
\displaystyle
\frac{\sin(\alpha+\gamma)}
{2\sin\gamma\;\sin\frac{\alpha+\gamma+\theta}{2}\;\sin\frac{\alpha+\gamma-\theta}{2}},
& \theta_{<}<\gamma<\theta_{>},\ \ \theta\le\frac{\pi}{2}
\ \ \text{(single move } \beta-\theta\text{)},\\[5mm]
\displaystyle
\frac{\sin(\gamma-\alpha)}
{2\sin\gamma\;\sin\frac{\alpha-\gamma+\theta}{2}\;\sin\frac{\alpha-\theta-\gamma}{2}},
& \theta_{<}<\gamma<\theta_{>},\ \ \theta>\frac{\pi}{2}
\ \ \text{(single move } \beta+\theta\text{)},\\[5mm]
\displaystyle
\frac{1}{\cos(\alpha+\theta)-\cos\gamma},
& 0<\gamma<\theta_{<} \ \ \text{(double region)}.
\end{cases}
\end{align}
\end{widetext}
The rows mirror those of the upper half-plane under $\beta\mapsto2\pi-\beta$.
In the two-barrier examples of the main text, $\theta=\pi/3\le\pi/2$, and
the poles that occur lie in the death and single-move sectors only.

\bibliography{ref}

\end{document}